\newif\ifreport\reportfalse
\documentclass[conference]{IEEEtran}

\usepackage{amsmath}
\usepackage{amsthm}
\usepackage{amsfonts} 
\usepackage{float}
\usepackage{graphicx}
\usepackage{stfloats}
\usepackage{subfig}
\usepackage{float}
\usepackage{bbm}
\usepackage{cases}
\usepackage[linesnumbered,ruled]{algorithm2e}

\usepackage{booktabs}
\usepackage{cite}
\usepackage{caption}
\usepackage{multirow}
\usepackage{amsmath}
\usepackage{amsthm}
\usepackage{amsfonts} 
\usepackage{float}
\usepackage{graphicx}
\usepackage{stfloats}
\usepackage{subfig}
\usepackage{float}
\usepackage{bbm}
\usepackage{amssymb}
\usepackage{xcolor}
\usepackage[short]{optidef}
\usepackage{amssymb}
\usepackage{bbm}

\newtheorem{theorem}{Theorem}
\newtheorem{lemma}{Lemma}
\newtheorem{corollary}{Corollary}
\newtheorem{proposition}{Proposition}
\DeclareMathOperator*{\argmin}{arg\,min}

\begin{document}

\title{Age Minimization with Energy and Distortion Constraints}

\author{\IEEEauthorblockN{Guidan Yao}
\IEEEauthorblockA{Department of ECE\\
The Ohio State University\\
Columbus, OH, USA\\
Email: yao.539@osu.edu}
\and
\IEEEauthorblockN{Chih-Chun Wang}
\IEEEauthorblockA{Elmore Family School of ECE\\
Purdue University\\
West Lafayett, IN, USA\\
Email: chihw@purdue.edu}
\and
\IEEEauthorblockN{Ness B. Shroff}
\IEEEauthorblockA{Department of ECE and CSE\\
The Ohio State University\\
Columbus, OH, USA\\
Email: shroff.11@osu.edu}}


\maketitle

\begin{abstract} 
In this paper, we consider a status update system, where an access point collects measurements from multiple sensors that monitor a common physical process, fuses them, and transmits the aggregated sample to the destination over an erasure channel. Under a typical information fusion scheme, the distortion of the fused sample is inversely proportional to the number of measurements received. Our goal is to minimize the long-term average age while satisfying the average energy and general age-based distortion requirements. Specifically, we focus on the setting in which the distortion requirement is stricter when the age of the update is older. We show that the optimal policy is a mixture of two stationary, deterministic, threshold-based policies, each of which is optimal for a parameterized problem that aims to minimize the weighted sum of the age and energy under the distortion constraint. We then derive  analytically the associated optimal average age-cost function and characterize its performance in the \emph{large threshold regime}, the results of which shed critical insights on the tradeoff among age, energy, and the distortion of the samples. We have also developed a closed-form solution for the special case when the distortion requirement is independent of the age, arguably the most important setting for practical applications. 
\end{abstract}




\maketitle

\section{Introduction}
For status update systems, it is important that the destination
receives fresh updates. However, a traditional metric like delay cannot fully characterize the freshness of information updates. For example, if the information is updated infrequently, then the updates are not fresh even though the delay is small. To this end, the age of information or simply the age was introduced in \cite{kaul2011minimizing} as a metric to represent the freshness (more precisely the staleness) of an update that simultaneously take into the update frequency and the delay into a single metric. 

\subsection{Problem and Applications}
In this paper, we consider a status update system, in which an access point receives measurements from multiple sensors, fuses them and transmits the aggregated sample to a remote monitor over wireless erasure channels.

Examples of the system can be found in wireless sensor networks (WSNs) and IoT systems. In certain IoT or WSNs applications like smart camera networks \cite{wang2013intelligent} or healthcare applications \cite{naresh2020internet}, multiple nodes (IoT devices/sensors) are used to observe a \emph{common} physical process. In ~\cite{kalor2019minimizing,zhou2020age,shao2021partially,9155238,tsai2021unifying}, the authors took this scenario into account in age relevant problems.
In addition, in wireless sensor networks or IoT systems, instead of allowing all nodes to directly communicate with the receiver, one node may be selected as a gateway/relay to forward collected data to the receiver in order to reduce the energy consumption \cite{gupta2005cluster}. Specifically, we have two examples as follows: 

\emph{Example 1: healthcare.} In healthcare architectures \cite{abdelmoneem2019cloud,kong2016design}, a sink node like a mobile device or a smart watch collects health indicators from wearable biomedical and activity sensors including ECG
collection, blood pressure, blood oxygen. Then, the collected data is sent to a cloud or back-end server for further processing/analysis. 

\emph{Example 2: smart agriculture.} In one mode of smart agriculture \cite{ayaz2019internet}, sensor nodes in a mesh network collect and transmit data first to the gateway, a designated node in the mesh network. Then, the gateway forwards this data to the farm management system using the WAN network.


Since wireless channels are not reliable and different sensor nodes may have their own sleep-wake schedules, the access point usually receives a random number of measurements at each time slot. In this work, we assume the sensors nodes (in different positions) observe a common process from different views. Under this assumption, the larger the number of received measurements, the higher the quality of the collected data in each time instant, and the less distortion of the update for the physical process. To ensure that precious resources are only used to forward samples of high quality, we impose a {\em distortion requirement} such that  at each time slot the access point can forward the fused/aggregated sample {\em only if} the number of received measurements is no less than a predefined threshold, which thus guarantees low distortion of each update.

In addition to jointly considering age and distortion, we note that nodes in status update systems are usually battery-powered and thus energy limited. See the two examples discussed earlier. Since communication energy/cost savings is a critical design consideration of any IoT device schedulers, the goal of this paper is to minimize the long-term average age under a long-term energy constraint, while respecting the aforementioned distortion requirement for each update.

\subsection{Related Works}
\begin{table*}[]
\footnotesize
\caption{Related Works}
\small
\label{Related_Works}
\begin{tabular}{lllll}
\hline\hline
Ref.        & Goal                                         & \begin{tabular}[c]{@{}l@{}}Energy \\ Constraint\end{tabular}                             & \begin{tabular}[c]{@{}l@{}}Distortion\\ Requirement\end{tabular}               & \begin{tabular}[c]{@{}l@{}}Channel\\ State Information\end{tabular}              \\ \hline
{\cite{gu2019timely}}          & \begin{tabular}[c]{@{}l@{}@{}}Minimize the average age by \\ optimizing the transmit power and the \\ maximum allowable transmission times \end{tabular} & \begin{tabular}[c]{@{}l@{}}Average power \\ consumption \end{tabular} 
& No                                                                            & \begin{tabular}[c]{@{}l@{}@{}}Error-prone \\ channel with \\fixed probability\end{tabular}          \\ \hline

{\cite{huang2021age}}          & \begin{tabular}[c]{@{}l@{}}Minimize the weighted sum of \\ the age and total energy consumption\end{tabular} & \begin{tabular}[c]{@{}l@{}}Limit number of \\ retransmissions \end{tabular}                                                            & No                                                                            & \begin{tabular}[c]{@{}l@{}}Distribution information\\ or unknown\end{tabular}          \\ \hline
{\cite{9736576, yao2021age}} & Minimize long-term average age                                                                                                        & \begin{tabular}[c]{@{}l@{}}Average energy \\ consumption \end{tabular}                                                              & No                                                                            & \begin{tabular}[c]{@{}l@{}@{}}Fading channel\\ imperfect channel \\ information\end{tabular}               \\ \hline
{\cite{feng2021age}}          & \begin{tabular}[c]{@{}l@{}} Minimize long-term average age\end{tabular}         & \begin{tabular}[c]{@{}l@{}@{}}Energy causality\\ constraint at the \\ EH sensor\end{tabular} & No                                                                            & \begin{tabular}[c]{@{}l@{}}Erasure channel with \\ fixed error probability\end{tabular} \\ \hline

{\cite{zheng2020age}}          & \begin{tabular}[c]{@{}l@{}} Study age-energy region by \\ studying average age minimization\end{tabular}         & \begin{tabular}[c]{@{}l@{}}Limited by\\ harvested energy\end{tabular} & No                                                                            & Fixed noise power \\ \hline

{\cite{rajaraman2021not}}          & \begin{tabular}[c]{@{}l@{}}Minimize the weighted sum of the\\ age, distortion and energy\end{tabular}                                & No                                                                                       & \begin{tabular}[c]{@{}l@{}}Soft constraint on controllable\\ distortion due to compression           \end{tabular}  &           No transmission failure                                                                       \\ \hline
{\cite{dong2020energy}}           & \begin{tabular}[c]{@{}l@{}}Minimize the weighted sum of the\\ age and distortion caused by noise\end{tabular}                                       & \begin{tabular}[c]{@{}l@{}}Limited by \\ harvested energy           \end{tabular}                                                          & \begin{tabular}[c]{@{}l@{}}Soft constraint on controllable\\ distortion due to observation noise           \end{tabular}                                              &  Gaussian channel                                                                                \\ \hline
{\cite{bastopcu2019age,bastopcu2021age}}    & Minimize the average age over a time $T$                                                                                             & No                                                                                      & \begin{tabular}[c]{@{}l@{}}Hard constraint on controllable \\ distortion due to processing time          \end{tabular}        & No transmission failure                                                          \\ \hline\hline
\end{tabular}
\end{table*}

There is a large body of literature that studies the trade-off between the age and energy. In \cite{gu2019timely}, the authors studied the trade-off in an IoT system by controlling the allowable times of retransmissions. \cite{huang2021age} studied power control policies that minimize the weighted sum of the age and energy consumption (including sensing and transmission energy costs) with constraint on times of retransmissions. \cite{yao2021age} and \cite{9736576} studied transmission scheduling over time-correlated fading channel to minimize long-term average age under an energy constraint. In \cite{zheng2020age}, the authors investigated the trade-off between the age and the storable energy at the IoT device in a wireless powered communication network. The authors in \cite{feng2021age} designed optimal online status updating policy to minimize the long-term average age at the destination, subject to the energy causality constraint at an energy-harvesting sensor. 

Different from the above works, one key consideration of this work is the focus on the {\em distortion requirement}. This requirement is hard in the sense that it has to be met all the time. In contrast, some papers deal with a soft distortion requirement ~\cite{dong2020energy,rajaraman2021not}, which assume there exists a trade-off between the age and distortion. Although both requirements have applications, in practical system, it can be more difficult to satisfy requirements from different aspects of the system simultaneously. As will be seen, our paper directly links the distortion to the (random) number of received measurements at any time slot. In general, distortion may be caused by other sources as well. For example, the distortion considered in ~\cite{hu2020balancing,rajaraman2021not} is caused by by compression. In particular, \cite{rajaraman2021not} studied a scheduling problem which aims to minimize the weighted sum of the age, distortion and energy, where distortion is determined by the number of bits sent for each source.
Paper \cite{hu2020balancing} investigated the trade-off between the age and the distortion caused by compression via assigning compression bits to packets in the queue and transmission scheduling. Distortion may also be caused by observation noise. In \cite{dong2020energy}, the authors considered this type of distortion and studied the optimal power control policy that minimizes the weighted sum of the age and distortion. 
In ~\cite{bastopcu2019age,bastopcu2021age}, the authors considered the distortion caused by the processing time and studied age-optimal distortion constrained updating policies. Specifically, \cite{bastopcu2019age} considered a fixed distortion requirement while \cite{bastopcu2021age} considered an age-dependent distortion requirement. For comparison, we summarize the related works in Table \ref{Related_Works}. 

\subsection{Key Contributions}
In the paper, we focus on the setting for which the distortion requirement of each update is stricter when the age of the system is older.
The idea is that if the age is also a source of distortion (a reasonable assumption, since the larger the age, the more likely that the estimates are poorer), then as the age increases, we would like to place a stricter allowable allowable distortion criterion (i.e., larger lower bound) of new updates to ensure that the overall quality of transmission is maintained. 
We then develop the optimal transmission control policy that minimizes the long-term average age under a long-term energy constraint while respecting the given age-based distortion requirement on each update.
Our key contributions are as follows:
\begin{itemize}
    \item We investigate the trade-offs among the age, energy and the distortion. Under our setting, we show that the optimal policy, which minimizes the long-term average age with the energy and distortion requirements, is a mixture of two stationary deterministic policies (Theorem \ref{Ori_Lag}). We also show that each stationary deterministic policy is optimal for a parameterized average cost problem, which aims to minimize the weighted sum of the age and energy while respecting the given age-based distortion requirement (Theorem \ref{Ori_Lag}). Further, we prove that the policy is of a threshold-type, i.e., a transmission is scheduled if (i) the age exceeds a certain threshold, and (ii) the distortion requirement is met (Theorem \ref{opt_stru}).
    \item We derive the average cost of the parameterized average cost problem (Theorem \ref{new}), and prove that it is a piecewise function of the earlier mentioned threshold (Theorem \ref{cost_cal}), and analytically characterize the property of the function (Theorem \ref{new}). By leveraging Theorems \ref{cost_cal} and \ref{new}, we circumvent the difficulty in dealing with an infinite state space when using classical solutions like the Relative Value Iteration (RVI). This allows us to develop low-complexity algorithms for both parameterized average cost problem (Algorithm \ref{alg0}) and the original problem (Algorithm \ref{alg1}).
    \item  In additional to characterizing the optimal policy for the general setting, we consider a special case of the parameterized average cost problem, where the distortion requirement is a constant that is independent of the age, arguable the most important setting for practical applications. In this important but simpler setting, we obtain a closed form expression for the optimal threshold (Corollary \ref{cor: special_case}), which allows us to examine the relationship among (a) the transmission threshold; (b) the probability of meeting the distortion requirement; and (c) the erasure probability of the access point's transmission. Specifically, we show that (i) the optimal threshold increases with the probability that the distortion requirement is met; (ii) when the energy is the dominant issue, the optimal threshold increases with the error probability of transmission from the access point to destination (Theorem \ref{lem:special_case}). But the situation reverses (from being an increasing function of the error probability to being a decreasing one) if the age is the dominant issue.
\end{itemize}


\section{System Model}
\label{scheduler}
We consider a status update system, in which an access point receives measurements from multiple sensors, fuses them, and then transmits the aggregated sample to a remote monitor/receiver, as shown in Fig. \ref{sysmodel}.
We consider a time-slotted system and use $t\in\{1,2,\cdots\}$ as the time index. At the beginning of each time slot, $M$ sensors measure the same physical process from their own perspectives and transmit the measurements to the access point over the wireless channels. Then, the access point decides whether to transmit an {\em update} to the remote monitor. Here an update can mean sending the entire set of received measurements to the remote monitor for further processing or it could mean sending the aggregated sample after fusing the data locally. The access point's action is denoted by $u_{t} \in \mathcal{U}\triangleq\{0,1\}$, where $u_{t}=1$ means transmission, and $u_{t}=0$ denotes forfeiting the transmission for time slot $t$.

We assume erasure channels. Specifically, the erasure probability of a transmission from the sensor $m$ to the access point is $q_m$, for $m\in\{1,2,\cdots,M\}$, and the erasure probability of the transmission from the access point to the receiver is $p$. 
\ifreport
\begin{figure}[h]
    \centering
    \includegraphics[width=0.6\textwidth]{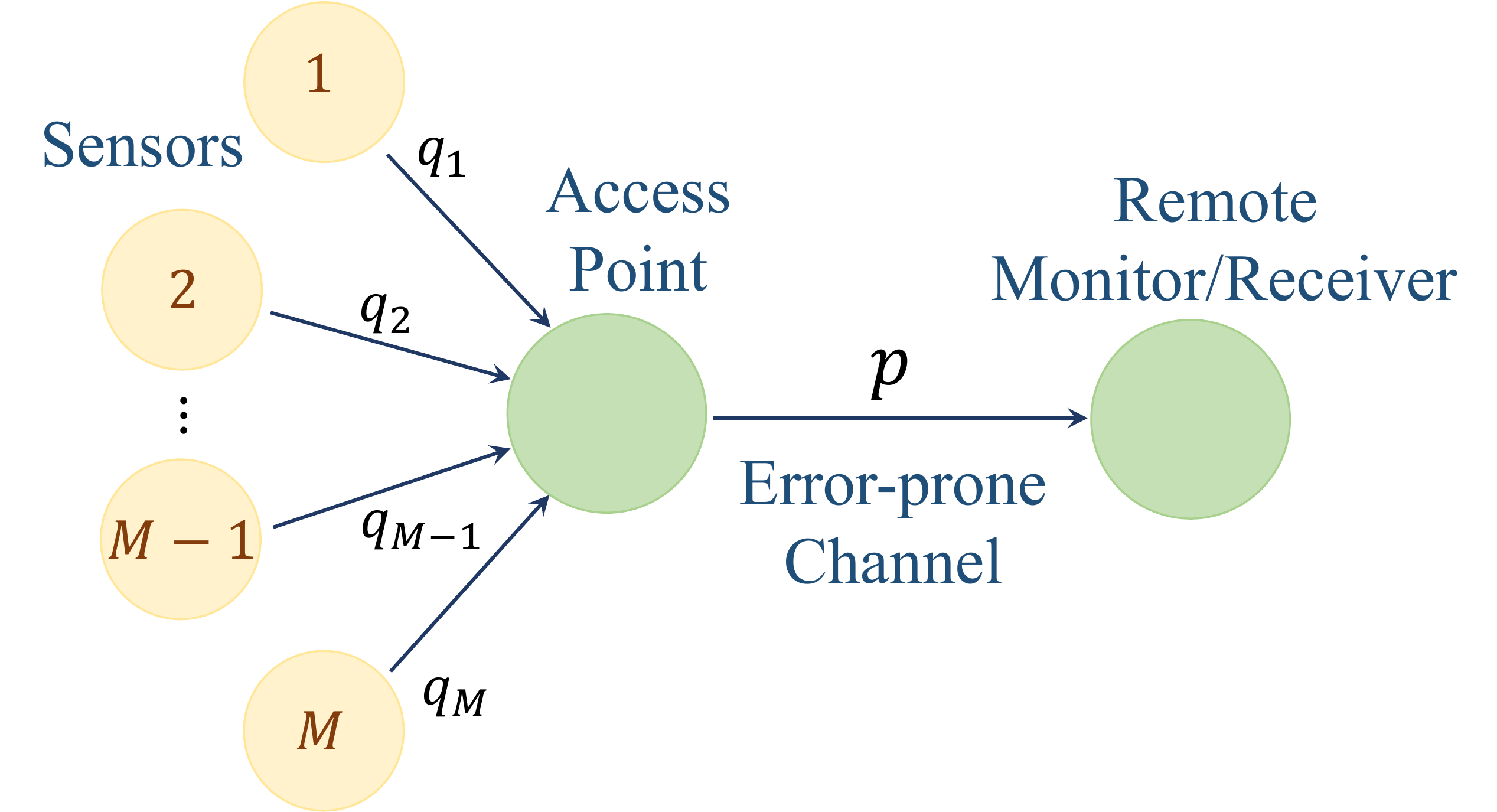}
    \caption{System Model}
    \label{sysmodel}
\end{figure}
\else 
\begin{figure}[h]
    \centering
    \includegraphics[width=0.32\textwidth]{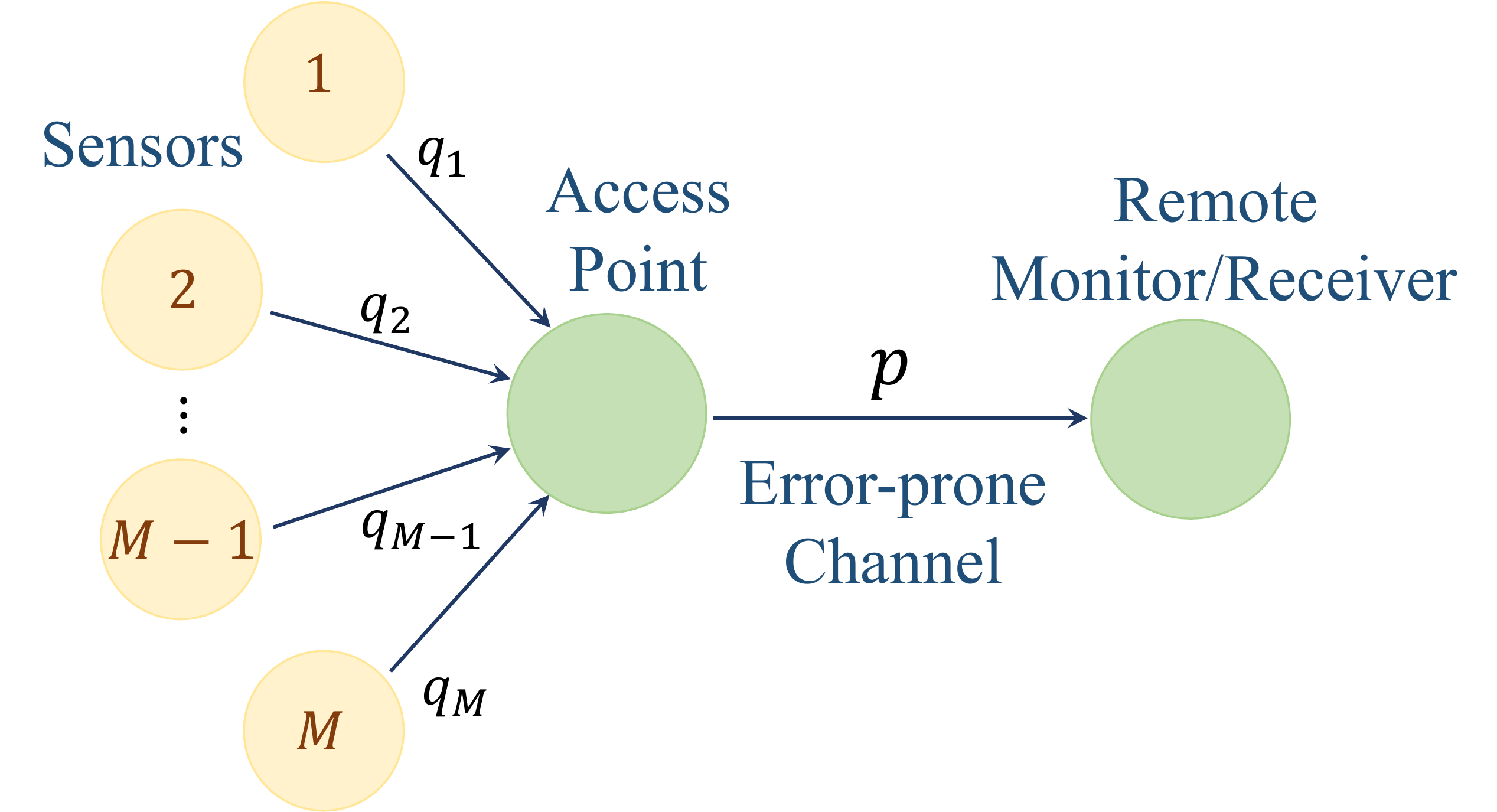}
    \caption{System Model}
    \label{sysmodel}
\end{figure}
\fi

\subsection{Age of Information}
Age of information (AoI), or simply the age, reflects the timeliness of the information at the remote monitor/receiver. It is defined as the time elapsed since the generation of the most recently received update sample at the receiver. Let $\Delta_{t}$ denote the age at the beginning of the time slot $t$. Let $U(t)$ denote the generation time of the last successfully received status update at time $t$. Then, $\Delta_{t}$ is given by $\Delta_{t}\triangleq t-U(t)$, which can be iteratively computed by
\begin{align}
\label{age_update}
& \Delta_{t+1}=
  \begin{cases}
  1 & \text{if transmission is successful}, \\  
  1+\Delta_{t} & \text{otherwise}. 
  \end{cases} 
\end{align}
In this work, we assume $\Delta_1=1$ for initialization.
\subsection{Distortion requirement and energy constraints}
\label{sec: const}
Since transmissions from sensors to the access point are not reliable, we use the random variable $\Lambda_t\in\{0,1,\cdots, M\}$ to denote the number of received measurements from $M$ sensors at time slot $t$. After fusing $\Lambda_t$ measurements to an aggregated sample, the corresponding distortion is a monotonically decreasing\footnote{In this work, the terms {\em decreasing} and {\em non-increasing} are considered interchangeable.} function $F_\mathrm{dist}(\Lambda_t)$ of $\Lambda_t$. To guarantee the quality of each update, we suspend access-point transmission whenever 
\begin{align}
F_\mathrm{dist}(\Lambda_t)> F_\mathrm{thre}(\Delta_t),\label{eq:dist-111}
\end{align}
i.e., we prohibit the access point from transmitting a low-quality (high-distortion) sample since it is essentially a waste of resources. We allow the threshold $F_\mathrm{thre}$ to depend on the age at the receiver. While \eqref{eq:dist-111} has a clear physical meaning, we can simplify {\em distortion requirement} to: we always choose $u_t=0$, if 
\begin{align}
&\Lambda_t< D(\Delta_t)\triangleq F_\mathrm{dist}^{-1}( F_\mathrm{thre}(\Delta_t)). \label{eq:dist-222}
\end{align}
The $D(\cdot)$ function is the concatenation of $F_{\mathrm{dist}}^{-1}()$ and $F_\mathrm{thre}()$, and we call it the {\em distortion function} that maps the age to the number of received measurements below which the access point will always discard the measurements ($u_t=0$) due to the lack of fidelity (distortion being too high). 

We assume that $D(\cdot)$ is an increasing function of the age on $[1,\infty)$ (since the smallest age is 1). Since $\Lambda_t$ in condition ~\eqref{eq:dist-222} is an integer between $0$ and $M$, without loss of generality, we assume $D(\cdot)$ is a piecewise constant function satisfying:
\begin{align}
D(\Delta)=h_l, \text{\ \ if\ \ } \Delta\in[\delta_l,\delta_{l+1}),~\forall l\in[1,L].
\end{align}
Namely, if $\Delta$ falls into the interval $[\delta_l,\delta_{l+1})$, then $D(\Delta)=h_l$. Here we assume $\delta_1=1$, $\delta_{L+1}=\infty$ and $1\leq h_l<h_{l+1}\leq M$ for all $l<L$, since if $h_L> M$, the transmission will never be made after the age exceeds $\delta_L$ (condition \eqref{eq:dist-222} will always hold then).
Fig.~\ref{fig:distfun_inc} provides an example of the distortion function, where $L=3$ and $M=8$.
\ifreport
\begin{figure}
    \centering
    \includegraphics[width=0.55\textwidth]{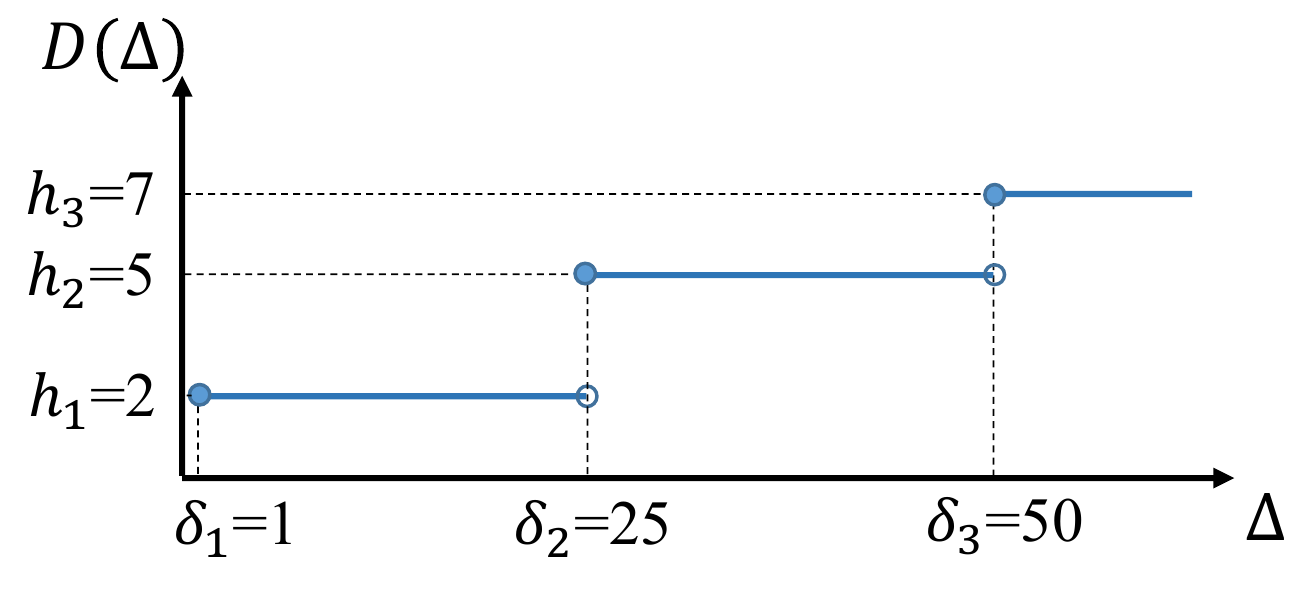}
    \caption{Example of distortion function}
    \label{fig:distfun_inc}
\end{figure}
\else
\begin{figure}
    \centering
    \includegraphics[width=0.38\textwidth]{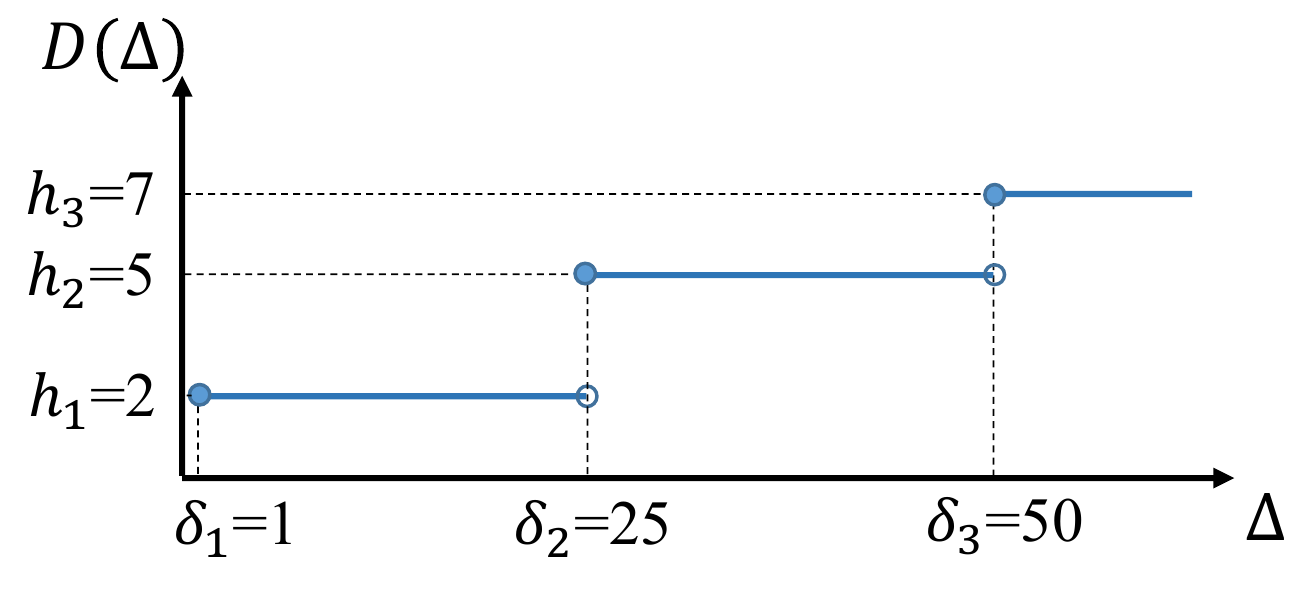}
    \caption{Example of distortion function}
    \vspace{-0.5cm}
    \label{fig:distfun_inc}
\end{figure}
\fi 



The access point consumes energy for each transmission. We assume that each transmission consumes the same energy which is normalized as one unit energy, a setting similar to \cite{tsai2021jointly}. To avoid excessive energy consumption, we employ a long-term average energy consumption constraint for the access point, which will be formalized later in \eqref{avg_energy}. Note that in the system, sensors take measurements periodically, and their energy consumption is fixed and thus not included in our optimization problem.


\subsection{Optimization Problem}
Our objective is to design a transmission control policy $\pi$ that minimizes the following long-term average age
\begin{align}
	&\bar{A}(\pi)\triangleq \lim_{T \rightarrow \infty} \frac{1}{T}\mathbb{E_\pi}\big[\sum_{t=1}^T \Delta_t\big],
	\label{avg_age}
\end{align}
while the long-term average energy consumption $\bar{E}(\pi)$ must not exceed $E_{\text{max}}\in (0,1]$, i.e.
\begin{align}
	&\bar{E}(\pi)\triangleq \lim_{T \rightarrow \infty} \frac{1}{T}\mathbb{E_\pi}\big[\sum_{t=1}^T u_t\big]\leq E_\text{max},
	\label{avg_energy}
\end{align}
and the distortion requirement in \eqref{eq:dist-222} is satisfied for all $t$, i.e.
\begin{align}
    \Lambda_t\geq u_t D(\Delta_t),\ \  \forall t\geq 1,\label{dist_req}
\end{align}
where $\mathbb{E}_\pi$ denotes expectation under policy $\pi$.

\section{Constrained MDP Formulation and Lagrangian Relaxation}
\label{sec: formulation}
\subsection{Constrained MDP Formulation}
\label{components}
The optimization problem can be formulated as a constrained MDP. 

\textbf{States:} The system state consists of the age and the number of received measurements at time $t$, i.e., $\mathbf{s}_t=(\Delta_t,\Lambda_t)$. Clearly, the state space, denoted by $\mathcal{S}\triangleq \{(\Delta,\Lambda): \Delta\in \mathbb{N}^+,  \Lambda\in\{0,1,\cdots,M\}\}$ is countably infinite. 

\textbf{Actions:} Action set is $\mathcal{U}=\{0,1\}$ as defined in Section \ref{scheduler}. Here we directly embed the distortion requirement in \eqref{dist_req} in the setting by assigning a heterogeneous action set for each $\mathbf{s}_t$. That is, define $A_{\mathbf{s}_t}\triangleq \{u_t\in \mathcal{U}: \Lambda_t\geq u_t D(\Delta_t)\}$ as the admissible action set in state $\mathbf{s}_t$ satisfying the distortion requirement \eqref{dist_req}. 
For example, with the distortion requirement in Fig. \ref{fig:distfun_inc}, we have $A_{(5, 5)}=\{0,1\}$ while $A_{(51, 5)}=\{0\}$. 

\textbf{Transition Probability:} Given the current state $\mathbf{s}_{t}= (\Delta_{t}, \Lambda_{t})$ and action $u_{t}$ at time slot $t$, the transition probability to the state $\mathbf{s}_{t+1}=(\Delta_{t+1}, \Lambda_{t+1})$ at the next time slot $t+1$, which is denoted by $P_{\mathbf{s}_{t}\mathbf{s}_{t+1}}(u_t)$, is defined as
 \begin{align}
	P_{\mathbf{s}_{t}\mathbf{s}_{t+1}}(u_t)&\triangleq\mathbb{P}(\mathbf{s}_{t+1}|\mathbf{s}_{t},u_{t})=\mathbb{P}(\Delta_{t+1}|\Delta_t,u_t)P_\Lambda(\Lambda_{t+1}),
\end{align}
where
\begin{align}
&\mathbb{P}(\Delta_{t+1}|\Delta_t,u_t)=
  \begin{cases}
  p & \text{if}\, \, u_{t}=1,\Delta_{t+1}=1+\Delta_t, \\
  1-p & \text{if} \, \, u_{t}=1,\Delta_{t+1}=1,\\  
  1 &\text{if} \, \, u_{t}=0,\Delta_{t+1}=1+\Delta_t,\\
  0 & \text{otherwise},
  \end{cases}
\end{align}
and $P_\Lambda(\Lambda_{t+1})\triangleq\mathbb{P}(\Lambda\!=\!\Lambda_{t+1})$.

\textbf{Costs:} Given a state $\mathbf{s}_t=(\Delta_{t}, \Lambda_{t})$ and an action choice $u_t$ at time slot $t$, the cost of one slot is the age at the beginning of this slot, i.e., we have
\begin{equation}
	C_\Delta(\mathbf{s}_t,u_{t})=\Delta_{t}.
\end{equation}
Moreover, the energy consumption of one slot is
\begin{equation}
	C_E(\mathbf{s}_t,u_{t})=u_t.
\end{equation}

Let $\Pi$ denote the set of \emph{feasible} policies that satisfy the distortion requirement, i.e., $u_t\in A_{\mathbf{s}_t}$, $\forall t$. Then, our control problem can be reformulated as a constrained Average-age MDP: 

\textit{Problem 1 (Average-age MDP):}
\begin{alignat}{2}
\label{ori_prob}
 \bar{A}^\star\triangleq   \min_{\pi\in\Pi} \quad & \bar{A}(\pi)= \limsup_{T \rightarrow \infty} \frac{1}{T}\mathbb{E_\pi}\big[\sum_{t=1}^T C_\Delta(\mathbf{s}_{t},u_{t})\big] &  \\
    \mathrm{s.t.} \quad & \bar{E}(\pi)=\limsup_{T \rightarrow \infty} \frac{1}{T}\mathbb{E_\pi}\big[\sum_{t=1}^T C_E(\mathbf{s}_{t},u_{t})\big] \leq E_{\text{max}}, &   \nonumber
\end{alignat}
where $\bar{A}^\star$ is the optimal average age.
\subsection{Lagrange Relaxation of the Constrained MDP}
We now solve \eqref{ori_prob}. Given Lagrange multiplier $\beta$, the instantaneous Lagrangian cost at time slot $t$ is defined by
\begin{align}	C(\mathbf{s}_{t},u_{t};\beta)\triangleq C_\Delta(\mathbf{s}_{t},u_{t})+\beta C_E(\mathbf{s}_{t},u_{t}).
	\label{inst_lag_cost}
\end{align}
Then, the long-term average Lagrangian cost under policy $\pi$ is 
\begin{align}
	&\bar{L}(\pi;\beta)= \lim_{T \rightarrow \infty} \frac{1}{T}\mathbb{E_\pi}\big[\sum_{t=1}^T C(\mathbf{s}_{t},u_{t}; \beta)\big].\label{Lag_main}
\end{align}

We then solve the following average age-plus-cost MDP:

\textit{Problem 2 (Average age-plus-cost MDP):}
\begin{align}
	\bar{L}^\star(\beta)  \triangleq &\min_{\pi\in \Pi}  \bar{L}(\pi;\beta),
	\label{avg_cost}
\end{align}
where $\bar{L}^\star(\beta)$ is the optimal average Lagrangian cost with regard to $\beta$. 


\begin{theorem}
\label{Ori_Lag}
There exists a stationary randomized policy that solves the average-age MDP \eqref{ori_prob}. This policy can be expressed as a mixture of two stationary deterministic policies $\pi_{\beta^\star,1}$ and $\pi_{\beta^\star,2}$, where $\pi_{\beta^\star,1}$ and $\pi_{\beta^\star,2}$ differ in
at most a single state $\mathbf{s}$, and are both optimal policies for the average age-plus-cost MDP \eqref{avg_cost} given $\beta^\star$. The mixture policy $\pi_{\beta^\star}$ uses $\pi_{\beta^\star,1}$ with probability $\mu$ and $\pi_{\beta^\star,2}$ with probability $1-\mu$ in state $s$; it uses either policy (since they coincide) in other states, where $\mu\in [0,1]$
\end{theorem}
\begin{proof}
Please see Appendix \ref{proof_ori_lag}.
\end{proof}
\vspace{0.1cm}

\section{Solving the average age-plus-Cost MDP}
\label{sec: average_cost}
In this section, we investigate optimal policies that solve \eqref{avg_cost} and characterize its unique structure. We also give a more detailed analysis for the special case when the distortion function $D(\cdot)$ is constant (previously we assume $D(\cdot)$ is increasing). 

\subsection{Structure of the optimal policies}
\label{sec: main_res}
In this part, we first study the structure of the optimal policies, and then obtain the average cost function with regard to a threshold.

\begin{table*}[h]
\centering
\footnotesize
\caption{\small{Notations}}
\label{notations}
\begin{tabular}{l}
\hline\hline
$F(r)= \sum_{j=r}^{M}P_\Lambda(j), \ \ \forall \ 0\leq r\leq M$\\ \hline
$B_{l}=1-(1-p)F(h_l), \ \ \forall \ 1\leq l\leq L$\\ \hline
$w(i,j)=\mathbbm{1}_{\{i<j\}}\prod_{v=i}^{j-1}(B_{v})^{\delta_{v+1}-\delta_v}+\mathbbm{1}_{\{i\geq j\}}$
\\ \hline
 $J_{l}=-\frac{1}{(1-B_{{l}-1})^2}-\frac{\delta_{l}-1+\beta F(h_{{l}-1})}{1-B_{{l}-1}}$\\
 \ \ \ \ \ \ $+\sum_{j={l}}^{L}w({l},j)\Big(\frac{1-\mathbbm{1}_{\{j<L\}}(B_{j})^{\delta_{j+1}-\delta_j}\big(1+(\delta_{j+1}-1+\beta F(h_j))(1-B_j)\big)}{(1-B_j)^2}
    +\frac{\delta_j-1+\beta F(h_j)}{1-B_j}\Big), \ \ \forall \ 2\leq l\leq L$ 
 \\\hline $I_{l}=-\frac{1}{1-B_{{l}-1}}+\sum_{j={l}}^{L}\frac{1-\mathbbm{1}_{\{j<L\}}(B_{j})^{\delta_{j+1}-\delta_j}}{1-B_j}w(l,j), \ \ \forall \ 2\leq l\leq L$
 \\ \hline $O_{l}=\frac{1}{(1-B_{{l}-1})^2}+\frac{-1+\beta F(h_{{l}-1})}{1-B_{{l}-1}}, \ \ \forall \ 2\leq l\leq L+1$
 \\ \hline \hline
\end{tabular}
\vspace{-0.4cm}
\end{table*}

\begin{theorem}
\label{opt_stru}
For any fixed $\beta$, there exists an optimal age threshold $\Delta_\beta^\star$ such that the following policy 
\begin{align}
&u^\star(\Delta,\Lambda;\beta)=
  \begin{cases}
  1 & \text{if}\, \, \Delta\geq \Delta_{\beta}^\star \, \ \text{and}\, \ \Lambda\geq D(\Delta),\\
  0 & \text{otherwise},
  \end{cases} \label{optPol}
\end{align}
achieves the minimum average Lagrangian cost in \eqref{avg_cost}. 
\end{theorem}
\begin{proof}
Please see Section \ref{proof_thre1}.
\end{proof}
By Theorem \ref{opt_stru}, the optimal policies for \eqref{avg_cost} are of threshold-type in the age. That is, if $\Lambda_t\!\!<\!\!D(\Delta_t)\!$, then the access point is prohibited to transmit due to the distortion criterion in \eqref{eq:dist-222}. If  $\Lambda_t\geq D(\Delta_t)$ at time $t$,  the optimal policy would transmit if and only if the age $\Delta_t$ exceeds $\Delta_\beta^\star$. 

In Theorem \ref{cost_cal} below, we express the average Lagrangian cost as a function of any given threshold. 
\par Recall that $\delta_l$ is the leftmost point of the $l$-th interval in the distortion requirement function $D(\cdot)$, see Fig.~\ref{fig:distfun_inc} and Sec.~\ref{sec: const}. Given any $k\in\mathbb{N}^+$, define  
\begin{equation}
    l_k\triangleq \min\{l\leq L+1: \delta_l>k\}. \label{eqn:def_lk}
\end{equation}
Broadly speaking, $l_k$ is the inverse of $\delta_l$. The special definition in \eqref{eqn:def_lk} is because the given $D(\cdot)$ function (see Fig.~\ref{fig:distfun_inc}) may have a jump. Since, per our definition $\delta_{L+1}=\infty$, we always have $l_k\leq L+1$. We then have the following results.

\begin{theorem}
\label{cost_cal}
We use policy $\pi_k$ to denote the policy described in \eqref{optPol} but with an arbitrarily given threshold $k\in\mathbb{N}^+$. The long-term average Lagrangian cost of $\pi_k$ is expressed as:
\begin{footnotesize}
\begin{align}
   &\bar{L}(\pi_k;\beta) \notag\\
   =&\frac{0.5k^2-0.5k+(1-B_{{l_k}-1})^{-1}k
    +\mathbbm{1}_{\{l_k\leq L\}}J_{l_k} (B_{l_k-1})^{\delta_{l_k}-k}+O_{l_k}}{k-1+(1-B_{{l_k}-1})^{-1}
    +\mathbbm{1}_{\{l_k\leq L\}}I_{l_k} (B_{l_k-1})^{\delta_{l_k}-k}},\label{eqn:avg_cost}
\end{align}
\end{footnotesize}
for which the expressions of the deterministic sequences $\{B_{l}\}_{l=1}^L$, $\{J_{l}\}_{l=2}^L$, $\{I_{l}\}_{l=2}^L$, and $\{O_{l}\}_{l=2}^{L+1}$ are given in Table \ref{notations}.
\end{theorem}
\begin{proof}
Please see Appendix \ref{opt_thres0}.
\end{proof}

Even though the closed-form expression of $\bar{L}(\pi_k;\beta)$ greatly simplifies the problem, finding the optimal threshold $\Delta_\beta^\star$ through numerical search is still quite challenging since (i) the expression of \eqref{eqn:avg_cost} versus the threshold $k$ may exhibit complicated behavior (e.g., we have found some scenarios where \eqref{eqn:avg_cost} is neither convex nor concave) and (ii) the domain of $k$ is unbounded. Later in Theorem~\ref{new}, we prove that in a restricted sub-domain $k\geq \delta_L$, we can analytically find the best $k_\text{UB}$. Therefore, to search for the optimal $\Delta_\beta^\star$, we only need to exhaustively evaluate $\bar{L}(\pi_k;\beta)$ for all $k\in \{1,2,\cdots,\delta_L-1\}$, compare their values to $\bar{L}(\pi_{k_\text{UB}};\beta)$ in Theorems~\ref{cost_cal} and~\ref{new}, and then find the globally minimum $k^\star$. 

\begin{theorem}
\label{new}
Among the sub-domain $k\geq \delta_L$, the $k$ that leads to the smallest $\bar{L}(\pi_k;\beta)$ is given by
\begin{align}
    k_\text{UB} &\triangleq \argmin_{k\geq \delta_L}\bar{L}(\pi_k;\beta)=\max\{\delta_L,y\}, \label{eqn:k0}
\end{align}
where 
\begin{small}
\begin{align}
    y\!=\!\Bigg\lceil\!\!-\!\frac{1+B_L}{2(1-B_L)}\!+\!\left(\frac{B_L^2}{(1-B_L)^2}\!+\!\frac{B_L\!+\!2\beta F(h_L)}{1-B_L}\!+\!\frac{1}{4}\right)^{0.5}\!\Bigg \rceil.
\end{align}
\end{small}
\end{theorem}

\begin{proof}
Please see Appendix \ref{proof_min_ub}.
\end{proof}

\subsection{Special Case - Constant Distortion}
\label{sec:special_case}
Our results in Theorems~\ref{Ori_Lag} to \ref{new} hold for any increasing distortion requirement function $D(\cdot)$. We now consider a special case of constant $D(\Delta)=h\leq M$, $\forall \Delta$ (equivalently when $L=1$), arguably the most important scenario in practice since it says that the access-point only forwards the aggregated sample when its quality (distortion) meets a constant threshold. 

\begin{corollary}
\label{cor: special_case}
Given $\beta$ and constant distortion function $D(\Delta)=h$, the optimal scheduler of problem \eqref{avg_cost} is analytically described as follows:
\begin{align}
&u_{\text{cons}}^{\star}(\Delta,\Lambda;\beta)=
  \begin{cases}
  1 & \text{if}\, \,\Delta\geq \Delta_{\beta,\text{cons}}^\star\, \ \text{and}\, \ \Lambda\geq h,\\
  0 & \text{otherwise},
  \end{cases}
\end{align}
where the optimal threshold
$\Delta_{\beta,\text{cons}}^\star$ is given by 
\begin{footnotesize}
\begin{align}
    \Delta_{\beta,\text{cons}}^\star\!\! =
    \!\max\Bigg\{\!1,\!\Bigg\lceil\!\!-\!\frac{1+R}{2(1\!-\!R)}\!+\!\sqrt{ \!\frac{R^2}{(1\!-\!R)^2}\!+\!\frac{R\!+\!2\beta W}{1\!-\!R}\!+\!\frac{1}{4}}\Bigg \rceil\Bigg\},
    \label{opt_threval_spec}
\end{align}
\end{footnotesize}
where $W=\sum_{j=h}^{M}P_\Lambda(j)$ and $R=1-(1-p) W$.
\end{corollary}

\textbf{Remark:} By Corollary \ref{cor: special_case}, the optimal threshold $\Delta_{\beta,\text{cons}}^\star$ depends on the Lagrangian multiplier $\beta$, distortion requirement $h$, the distribution of the random number of received measurements $\Lambda$ and channel unreliability $p$. By simple algebraic simplification of \eqref{opt_threval_spec}, it is easy to show that when $\beta=0$, which corresponds to the case without energy constraint, we have $\Delta_{\beta,\text{cons}}^\star=1$. In addition, $\Delta_{\beta,\text{cons}}^\star$ is non-decreasing with $\beta$. This is because the increase of $\beta$ implies higher weights of the energy cost. To save energy cost, the threshold should be increased to reduce the transmission frequency.

To provide more insights on the special case with constant distortion, we investigate how the measurements arrival probability $W=\mathbb{P}(\Lambda\geq h)$ and the access-point-to-receiver erasure probability $p$ affect the optimal threshold $\Delta_{\beta,\text{cons}}^\star$ in Theorem \ref{lem:special_case}.

\begin{theorem}
\label{lem:special_case}
Given $\beta>0$ and $D(\Delta)=h$, we have the following properties: 

(i) If $\beta >\frac{1}{W}$, the optimal threshold $\Delta_{\beta,\text{cons}}^\star$ is increasing with respect to (w.r.t.) $p$; if $\beta <\frac{1}{W}$, the optimal threshold $\Delta_{\beta,\text{cons}}^\star$ is decreasing w.r.t.~$p$;

(ii) In fact, we can further strengthen the second half of (i) by the following: If $\beta <\frac{1}{W}$, the optimal threshold $\Delta_{\beta,\text{cons}}^\star=1$; 

(iii) The optimal threshold $\Delta_{\beta,\text{cons}}^\star$ is increasing w.r.t.~$W$.
\end{theorem}
\begin{proof}
Please see Appendix \ref{app: special_case}.
\end{proof}

\textbf{Remark}
Note $\frac{1}{W}$ {\em is the expected duration until the next time (slot) that the access point can receive  enough measurements to meet the distortion requirement.} Roughly speaking, if the access point does not send an aggregated update at this time, it will have to wait $\frac{1}{W}$ time slots before the next time it receives enough measurements to meet the distortion requirement. Therefore, $\frac{1}{W}$ can be viewed as the {\em cost of suspension}. Also recall that the Lagrangian multiplier $\beta$ can be viewed as the energy price of one transmission. Jointly, the intuition of the theorem can be explained as follows:
\begin{itemize}
    \item 
In the case $\beta >\frac{1}{W}$, the energy cost precedes the suspension cost. Then, as $p$ increases (channel worsens), it is better to reduce the transmission frequency in order to save energy while sacrificing slightly the age performance. As a result, the optimal threshold increases as stated in (i) of Theorem~\ref{lem:special_case}. On the other hand, if $\beta <\frac{1}{W}$, then the suspension cost dominates. The optimal threshold will decrease to ensure that we transmit more frequently to maintain a small average age (at the cost of increased energy consumption).
\item If $p=0$ (erasure probability is zero), we only need to consider the cost of sending one aggregated update. If we also have $\beta<\frac{1}{W}$, it means that the age cost outweighs the energy cost, which implies the optimal threshold is 1. Together with the second half  of (i), we will have (ii).
\item 
The intuition of (iii) is as follows. For any given $\beta$, $W$, the optimal threshold $\Delta^\star_{\beta,\text{cons}}$ would balance the {\em marginal} age cost and the {\em marginal} energy cost in \eqref{avg_cost} so that any perturbation of the threshold in either the positive or negative direction will decrease the performance. Consider a slightly larger $W'>W$. Since $\frac{1}{W}$ is the expected duration between two consecutive slots when the distortion requirement is met, the duration under $W'$ would be slightly smaller. Note that using the new $W’$ would increase the marginal energy cost (since we send more frequently) but decrease the marginal age cost (since we send more frequently). As a result, to re-balance the two marginal costs, the optimal policy would further increase $\Delta^\star_{\beta,\text{cons}}$ under the new $W'$.
\end{itemize}


\subsection{Proof of Theorem \ref{opt_stru}}
\label{proof_thre1}
A method to study average cost MDPs is to relate them to discounted cost MDPs. In this section, we (i) define discounted cost MDPs; (ii) obtain optimal policies for the discounted cost MDPs; and (iii) extend the results to average cost MDPs.

Given discount factor $\alpha \in (0, 1)$ and an initial state $\mathbf{s}$, the total expected discounted Lagrangian cost under a policy $\pi\in\Pi$ is 
\begin{align}
	&L_{\mathbf{s}}^{\alpha}(\pi;\beta)= \limsup_{T \rightarrow \infty} \mathbb{E}_{\pi}\big[\sum_{t=1}^T \alpha^{t-1}C(\mathbf{s}_{t},u_{t};\beta)|\mathbf{s}\big],
	\label{disc_cost0}
\end{align}

Then, the optimization problem of minimizing the total expected discounted Lagrangian cost can be cast as

\textit{Problem 3 (Discounted cost MDP):}
\begin{equation}
	V^{\alpha}(\mathbf{s})\triangleq\min_{\pi\in\Pi} \ L_{\mathbf{s}}^\alpha(\pi;\beta),
	\label{disc_cost_opt_00}
\end{equation}
where $V^{\alpha}(\mathbf{s})$ denotes the optimal total expected $\alpha$-discounted Lagrangian cost (for convenience, we omit $\beta$ in $V^{\alpha}(\mathbf{s})$).

We now introduce the optimality equation of $V^{\alpha}(\mathbf{s})$.
\begin{proposition}
\label{existence_discount_00}
	\noindent	(a) The optimal total expected $\alpha$-discounted Lagrangian cost, given by $V^{\alpha}(\Delta,\Lambda)$, satisfies the optimality equation as follows:
	\begin{align}
		V^{\alpha}\left(\Delta,\Lambda\right)=\min_{u\in A_{(\Delta,\Lambda)}}  Q^{\alpha}\left(\Delta,\Lambda;u\right), \label{opt_equ_disc_00}
	\end{align}	 
	where
	\begin{align}		Q^{\alpha}\left(\Delta,\Lambda;0\right)\!=&\Delta\!+\!\alpha \mathbb{E}V^{\alpha}\left(\Delta+1,\Lambda'\right);\label{q1_}\\
		Q^{\alpha}\left(\Delta,\Lambda;1\right)\!=&\Delta\!+\!\beta
		\!+\!\alpha \Big(p \mathbb{E}V^{\alpha}\!\left(\Delta\!+\!1, \Lambda'\right)\!+\!(1\!-\!p)\mathbb{E}V^{\alpha}\left(1,\Lambda'\right)\!\Big).\label{q2_}
	\end{align}
	
\noindent(b) A stationary deterministic policy determined by the right-hand-side of \eqref{opt_equ_disc_00} solves problem \eqref{disc_cost_opt_00}.
	
\noindent (c) Let $V_n^{\alpha}(\mathbf{s})$ be the cost-to-go function such that
$V_0^{\alpha}(\mathbf{s})\!=\!0$, for all $\mathbf{s}\in\mathcal{S}$ and for $n\geq 0$,
		\begin{align}
		V_{n+1}^{\alpha}(\Delta,\Lambda)=\min_{u\in A_{(\Delta,\Lambda)}} Q_{n+1}^{\alpha}(\Delta,\Lambda;u),
		\label{iteration00}
	    \end{align}
	    where
	    \begin{align}
		\!\!Q_{n+1}^{\alpha}\!\left(\Delta,\Lambda;0\right)\!=&\Delta\!+\!\alpha\mathbb{E} V_{n}^{\alpha}\left(\Delta+1,\Lambda'\right);\label{q1}\\
		\!\!Q_{n+1}^{\alpha}\!\left(\Delta,\Lambda;1\right)\!=&\Delta\!+\!\beta\!+\!\alpha \Big(p\mathbb{E} V_{n}^{\alpha}\left(\Delta\!+\!1, \Lambda'\!\right)\!+\!(1\!-\!p)\mathbb{E}V_{n}^{\alpha}		\left(1,\Lambda'\right)\!\Big).\label{q2}
	    \end{align}
	    Then, we have $V_n^{\alpha}(\mathbf{s}) \rightarrow V^{\alpha}(\mathbf{s})$ as $n\rightarrow \infty$, $\forall \mathbf{s}, \alpha$.
\end{proposition}
\begin{proof}
According to \cite{sennott1989average}, it suffices to show that there exists a stationary deterministic policy $f\in\Pi$ such that for all $\alpha$ and $\mathbf{s}$, we have $L_{\mathbf{s}}^\alpha(f;\beta)\!<\!\!\infty$.
Let $f$ be a policy that chooses $u=0$ at every time slot. Obviously, $f$ satisfies the distortion constraint and thus $f\in\Pi$. Further, given the initial state $\mathbf{s}_1=(\Delta,\Lambda)$, we have 
\begin{align}
	L_{\mathbf{s}_1}^\alpha(f;\beta)&= \limsup_{T \rightarrow \infty} \mathbb{E}_f\big[\sum_{t=1}^T \alpha^{t-1}\Delta_t|\mathbf{s}_{1}\big]\notag\\
	&=\sum_{n=0}^{\infty}\alpha^n(\Delta+n)\notag\\
	&=\frac{\Delta}{1-\alpha}+\frac{\alpha}{(1-\alpha)^2}<\infty.\notag
\end{align}
\end{proof}

Using the induction method in (c) of Proposition \ref{existence_discount_00}, we first show some properties of $V^{\alpha}(\mathbf{s})$ in Lemma \ref{property_00}.
\begin{lemma}
\label{property_00}
Given $\alpha$, the value function $V^{\alpha}(\mathbf{s})$ has properties:

(i) The value function $V^{\alpha}(\mathbf{s})$ is increasing w.r.t.~$\Delta$.

(ii) The value function $V^\alpha(\mathbf{s})$ is decreasing w.r.t.~$\Lambda$.
\end{lemma}
\begin{proof}
Please see Appendix \ref{proof_prop_dis_subprob}.
\end{proof}
\vspace{0.2cm}

Using the properties in Lemma \ref{property_00}, we further show that the optimal policies that solve discounted cost MDPs in \eqref{disc_cost_opt_00} are of threshold-type in Lemma \ref{disc_stru0}.
\begin{lemma}
\label{disc_stru0}
Given $\beta,\alpha$, the optimal policy that solves the discounted cost MDP \eqref{disc_cost_opt_00} is of threshold-type in the age. Specifically, there exists a threshold $\Delta_{\alpha,\beta}^{\star}$ such that it is optimal to transmit only when the age exceeds the threshold and the distortion requirement is met, i.e., $ \Delta\geq \Delta_{\alpha,\beta}^{\star} $ and $\Lambda\geq D(\Delta)$.
\end{lemma}
\begin{proof}
Please see Appendix \ref{proof_disc_stru}.
\end{proof}

 By \cite{sennott1989average}, under certain conditions (A proof of these conditions verification is provided in Appendix \ref{proof_cond}), the optimal policy for problem \eqref{avg_cost} can be viewed as a limit of a sequence of the optimal policies for the $\alpha$-discounted cost problems in \eqref{disc_cost_opt_00} as $\alpha\to 1$. Thus, there exist stationary deterministic policies that solve problem \eqref{avg_cost}, and the optimal policies are in the form \eqref{optPol}.

\section{Low-complexity Algorithm for average-age MDP}
\label{sec:alg_design}
In this section, we design a low-complexity algorithm for the average-age MDP. In particular, we first design a low-complexity algorithm to obtain the optimal policy for the average age-plus-cost MDP \eqref{avg_cost} given $\beta$, and then provide a way to determine optimal Lagrangian multiplier $\beta$.
\subsection{Optimal policy for average age-plus-cost MDP \eqref{avg_cost}}
In Theorem \ref{opt_stru}, we show that the optimal policy for \eqref{avg_cost} is of a threshold-type in the form of \eqref{optPol}. In order to obtain the optimal policy for \eqref{avg_cost}, it remains to obtain optimal threshold $\Delta_\beta^\star$. Using Theorems~\ref{cost_cal} and \ref{new}, we can find $\Delta_\beta^\star$ by the following Algorithm~\ref{alg0}. 
\vspace{-0.2cm}
\begin{algorithm}
\SetAlCapFnt{\footnotesize}
\caption{\footnotesize Optimal Threshold Calculator}
\label{alg0}    
\LinesNumbered
\footnotesize 
\textbf{Input:} $\{\delta_j\}_{j=1}^L$, $\{h_j\}_{j=1}^L$, $l=2$, $k^\star=\infty$, $\bar{L}^\star=\infty$, \;
\For{$k=1$ to $\delta_L$}{
\eIf{$k=\delta_L$}{
Calculate $k_\text{UB}$ using \eqref{eqn:k0} and $\bar{L}(\pi_{k_\text{UB}};\beta)$ using \eqref{eqn:avg_cost}\;
\If{$\bar{L}(\pi_{k_\text{UB}};\beta)<\bar{L}^\star$}{
$\bar{L}^\star=\bar{L}(\pi_{k_\text{UB}};\beta)$\, $k^\star=k_\text{UB}$\;}
}{
\If{$k=\delta_{l-1}$}{
Calculate $B_{l-1}$, $O_l$, $I_l$ and $J_l$ using Table \ref{notations}\;
$l=l+1$\;
}
Calculate $\bar{L}(\pi_k;\beta)$ using \eqref{eqn:avg_cost}\;
\If{$\bar{L}(\pi_k;\beta)<\bar{L}^\star$}{
$\bar{L}^\star=\bar{L}(\pi_{k};\beta)$\, $k^\star=k$\;}
}
}
\textbf{Output:} $k^\star$, $\bar{L}^\star$;
\end{algorithm}

\textbf{Remark:} Relative value iteration (RVI) is a classical way to solve the average cost MDP. However, RVI requires updating value functions of all states in each iteration. Since the state space is infinite in this problem, the state space would have to be truncated before applying RVI, which could introduce a significant error from the optimal solution. Compared with RVI, the advantage of the proposed algorithm \ref{alg0} is as follows:

(i) In Algorithm \ref{alg0}, we avoid dealing with infinite space. Instead, we compare average costs of threshold-type policies with \emph{finite} different thresholds.  

(ii) In terms of complexity, each iteration of RVI takes $O(MX)$, where $X$ is the bound for the age and should be larger than $\delta_L$ to reduce offset from the optimal solution. However, Algorithm \ref{alg0} takes $O(\delta_L+M)$ to find the optimal threshold /policy for \eqref{avg_cost}, which {\em has less complexity than even two iterations of RVI}.

\subsection{Lagrangian multiplier estimate}
In Theorem \ref{Ori_Lag}, we have shown that the optimal policy for the average-age MDP \eqref{ori_prob} is a mixture of two stationary deterministic policies $\pi_{\beta^\star,1}$ and $\pi_{\beta^\star,2}$, each of which is optimal for the average age-plus-cost MDP \eqref{avg_cost}. 

Our idea is to construct two sequences $\beta_n$ and $\beta_n'$ such that they satisfy
$\beta_n\leq \beta_n'$, $\beta_n \uparrow \beta^\star$ and $\beta'_n\downarrow \beta^\star$. By \cite{ceran2019average}, in practice, we can use optimal policies $\pi^\star_{\beta^-}$ and $\pi^\star_{\beta^+}$ to approximate policies $\pi_{\beta^\star,1}$ and $\pi_{\beta^\star,2}$, where $\beta^-\triangleq \beta_n$ and $\beta^+ \triangleq \beta'_n$ for reasonably large $n$; and $\pi^\star_{\beta^-}$ and $\pi^\star_{\beta^+}$ are optimal policies that solve \eqref{avg_cost} associated with $\beta^-$ and $\beta^+$, respectively. With $\beta^-$ and $\beta^+$, the randomization factor $\mu$ is calculated by
\begin{align}
    \mu=\frac{E_\text{max}-\bar{E}(\pi^\star_{\beta^+})}{\bar{E}(\pi^\star_{\beta^-})-\bar{E}(\pi^\star_{\beta^+})},
    \label{mu_cal}
    \vspace{-0.3cm}
\end{align}
where $\bar{E}(\pi)$ is the average energy cost under policy $\pi$ as defined in \eqref{avg_energy}. Then, the optimal policy for the average-age MDP chooses $\pi^\star_{\beta^-}$ with probability $\mu$ and $\pi^\star_{\beta^+}$ with probability $1-\mu$ after each successful delivery.

In Algorithm \ref{alg1}, we provide a way to obtain parameters $\beta^+$, $\beta^-$ and $\mu$, which determines the optimal policy for the average-age MDP \eqref{ori_prob}. In particular, we use the bisection method to obtain $\beta^-$ and $\beta^+$ that follows the methodology in \cite{sennott1993constrained}. 
In Algorithm~\ref{alg1}, the expression of the average energy cost with a threshold $k$ is 
\begin{small}
\begin{align}
    &\bar{E}(\pi_k)=\Bigg\{\frac{ F(h_{{l_k}-1})}{1-B_{{l_k}-1}}+\mathbbm{1}_{\{l_k\leq L\}}(B_{{l_k}-1})^{\delta_{l_k}-k}\bigg(-\frac{ F(h_{{l_k}-1})}{1-B_{{l_k}-1}}\notag\\
    &+\sum_{j=l_k}^{L}w(l_k,j) F(h_j)\frac{1-\mathbbm{1}_{\{j<L\}}(B_{j})^{\delta_{j+1}-\delta_j}}{1-B_j}\bigg)\Bigg\}z^{-1},
    \label{avg_cost_cal}
\end{align}
\end{small}

where $z=k-1+(1-B_{{l_k}-1})^{-1}+\mathbbm{1}_{\{l_k\leq L\}}I_{l_k} (B_{{l_k}-1})^{\delta_{l_k}-k}$ and $l_k$ is defined previously in \eqref{eqn:def_lk}. The derivation of \eqref{avg_cost_cal} is part of the proof of Theorem \ref{cost_cal}.

\begin{algorithm}
\SetAlCapFnt{\footnotesize}
\caption{\footnotesize Low-complexity Optimal Transmission Scheduler Estimator}
\label{alg1}    
\LinesNumbered
\footnotesize 
\textbf{Input:} $\epsilon>0$, $\beta^{-}$, $\beta^{+}$ \;
\While{$|\beta^+-\beta^-|>\epsilon$}{
$\beta=(\beta^++\beta^-)/2$\;
For the new $\beta$, obtain optimal threshold $k^\star$using Algorithm~\ref{alg0}\;
Calculate average energy cost using \eqref{avg_cost_cal}\;
 \eIf{$\bar{E}(\pi_{k^\star})>E_{\text{max}}$}{
 $\beta^-=\beta$\;}{
 $\beta^+=\beta$\;
 }
 }
 Compute $\mu$ using \eqref{mu_cal}\;
\textbf{Output:} $\mu$, $\beta^+$ and $\beta^-$\;
\end{algorithm}

\vspace{-0.4cm}
\section{Simulations}
\label{sec:simu}
In this section, we numerically evaluate the performance of
the proposed algorithms.
\subsection{Optimal threshold for problem \eqref{avg_cost}}
\label{sec:simu1}
To provide more insights, we investigate how the optimal threshold varies with $\beta$ and error probabilities, respectively. 
\subsubsection{Constant distortion requirement}
We first simulate the special case. In the simulation, we set $h=5$, $M=10$ and $q_m=q$, $\forall m$. Fig. \ref{fig:opt_thre_f} studies the optimal threshold versus $W$ given $p=0.5$ and $\beta \in\{5,10,20\}$, where $W$ is changed by changing $q$. From Fig. \ref{fig:opt_thre_f}, we observe that the optimal threshold increases with $W$. Fig. \ref{fig:opt_thre_p} studies the optimal threshold versus $p$ given $\beta=10$ and $W\in\{0.03,0.36,0.83\}$. From Fig. \ref{fig:opt_thre_p}, we observe that the optimal threshold increases with $p$ for $W=0.36$ or $W=0.83$ ($\beta>\frac{1}{W}$). Moreover, for all cases in Figs. \ref{fig:opt_thre_f} and \ref{fig:opt_thre_p} that satisfy $\beta<\frac{1}{W}$, the optimal threshold is one. These observations confirm our theoretical results in Theorem~\ref{lem:special_case}.
\begin{figure}
    \centering    \includegraphics[width=0.42\textwidth]{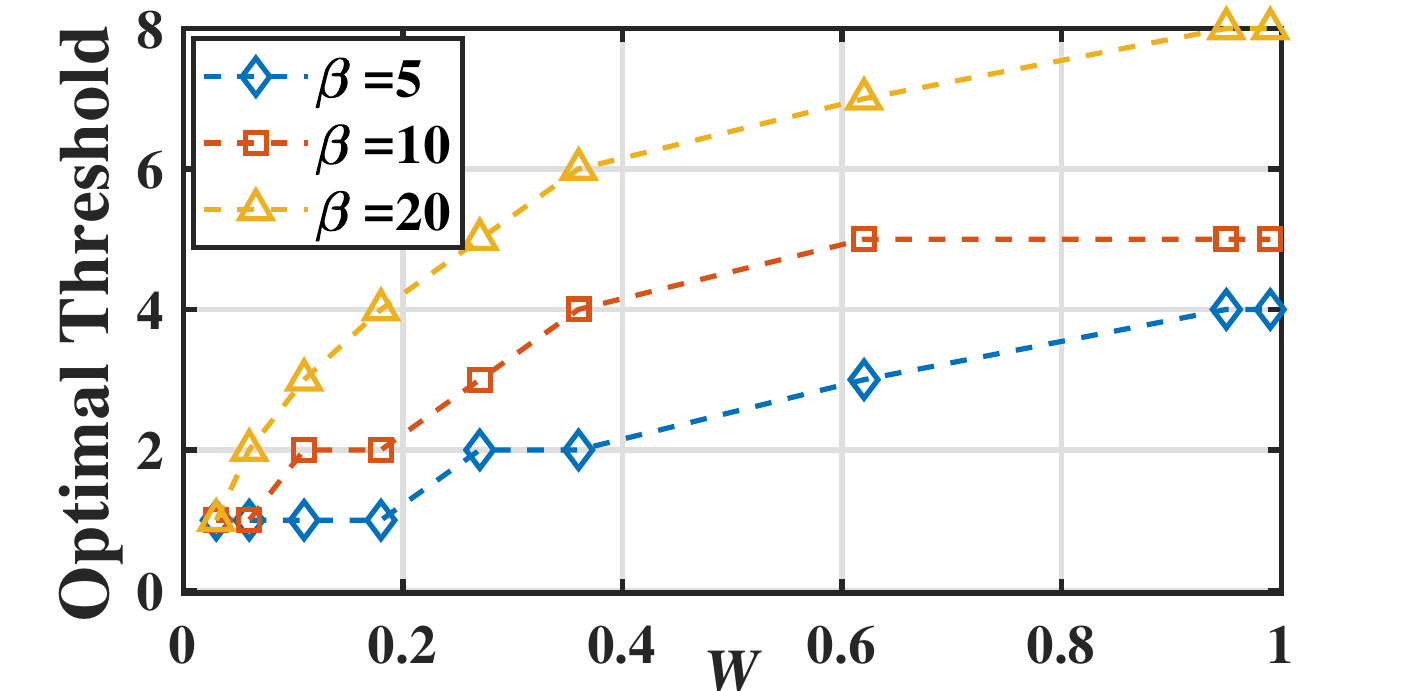}
    \caption{Optimal threshold $\Delta^\star_{\beta,\text{cons}}$ vs $W$ given $p=0.5$}
    \label{fig:opt_thre_f}
    \vspace{-0.5cm}
\end{figure}  

\begin{figure}
    \centering    \includegraphics[width=0.42\textwidth]{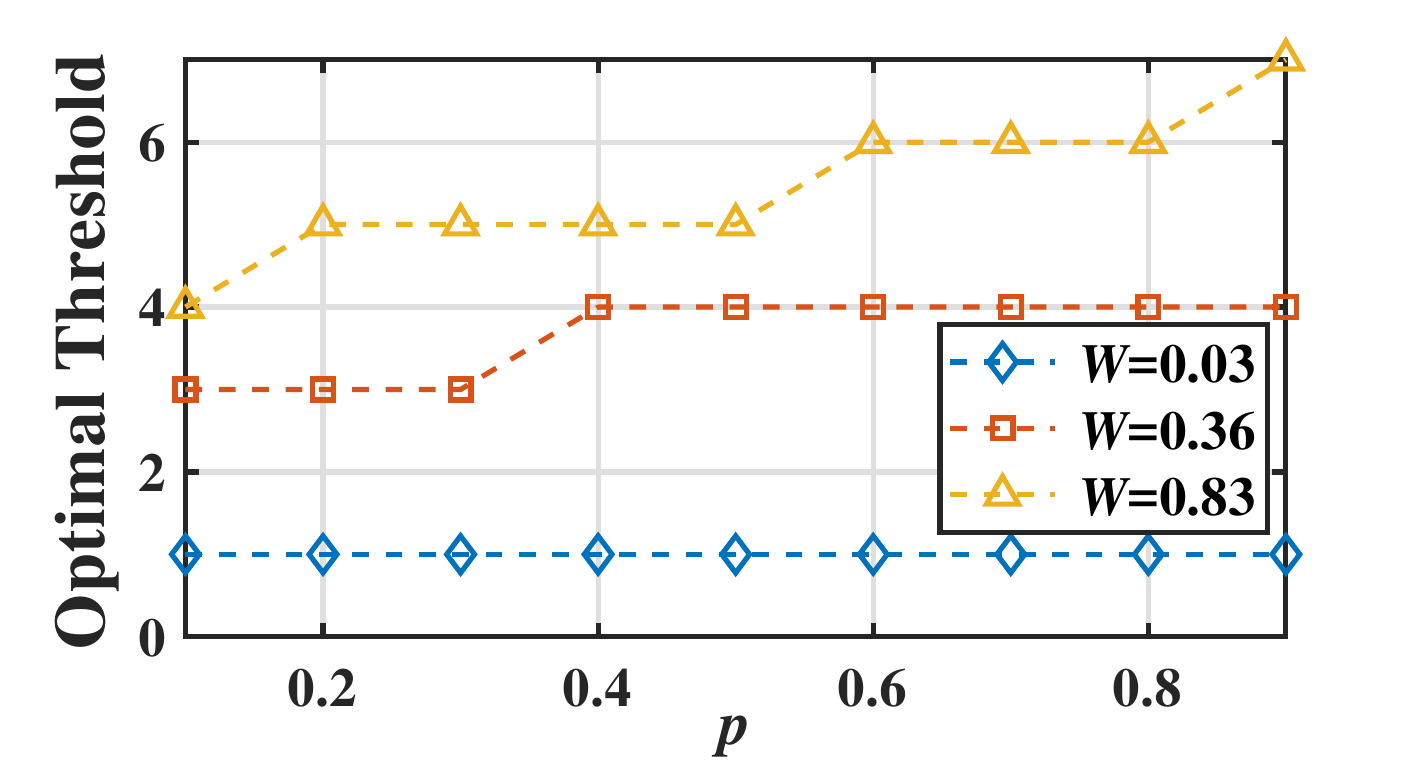}
    \caption{Optimal threshold $\Delta^\star_{\beta,\text{cons}}$ vs $p$ given $\beta=10$}
    \label{fig:opt_thre_p}
   \vspace{-0.7cm}
\end{figure} 

\subsubsection{Non-decreasing age-vs-distortion requirement}
For general non-decreasing distortion function, we set $M=8$, $L=3$, and $q_m=q$, $\forall m$. The distortion function is given by $\delta_1=1, \delta_2=25, \delta_3=50$ and $h_1=2, h_2=5, h_3=7$. Also see Fig.~\ref{fig:distfun_inc}. 

Under the above mentioned parameter settings, we obtain the optimal threshold for different $\beta\in \{5, 25, 45\}$, and error probabilities $p\in \{0.1, 0.2,\cdots 0.9\}$ and $q \in \{0.1, 0.2,\cdots 0.9\}$ using Algorithm \ref{alg0}.
The results are summarized in Fig. \ref{fig:OptThre}.
In each sub-figure of Fig. \ref{fig:OptThre}, we investigate the impact of $p$ and $q$ on the optimal threshold. On one side, we observe that the optimal threshold decreases with $q$. This is consistent with the analysis result in the constant $D(\cdot)$ case in Sec.~\ref{sec:special_case} even though we do not have any analytical proof of this phenomenon. The intuition is that the increase of $q$ will reduce the probability that the distortion requirement is satisfied. This increases the demand for more transmission opportunities to maintain low age, which reduces the optimal threshold. On the other side, we observe that the optimal threshold either increases or deceases or first increases and then decreases with $p$ (see $q=0.6$ in Figs.~\ref{fig:thre_beta5},~\ref{fig:thre_beta25}, and \ref{fig:thre_beta45}). Whether the optimal threshold increases with $p$ depends on whether the age or the energy cost is the dominant issue, also see our discussion right after Theorem~\ref{new}. In particular, when the dominant issue to deal with is the energy cost, the optimal threshold increases with $p$. This is because increasing $p$ implies increasing energy consumption for a successful update. To save energy, the optimal threshold is increased. When the dominant issue to deal with is the age, the optimal threshold decreases with $p$. This is because that the increase of $p$ implies that more transmission attempts are needed for a successful delivery. To keep the age low, the optimal threshold should be reduced to provide more transmission opportunities.

Moreover, comparing Figs. \ref{fig:thre_beta5}, \ref{fig:thre_beta25} and \ref{fig:thre_beta45}, we observe that the optimal threshold increases with $\beta$. This is because as $\beta$ increases, more weights are placed on the energy cost, which requires increasing the threshold to reduce the energy cost.
\ifreport
 \begin{figure*}[]
 \centering
 \subfloat[$\beta=5$]{
 \includegraphics[scale=.35]{Age_dist/OptThre_N8_beta5.eps}
 \label{fig:thre_beta5}
 }\ 
 \subfloat[$\beta=25$]{
 \includegraphics[scale=.35]{Age_dist/OptThre_N8_beta25.eps}
 \label{fig:thre_beta25}
 }\\  
 \subfloat[$\beta=45$]{
 \includegraphics[scale=.35]{Age_dist/OptThre_N8_beta45.eps}
 \label{fig:thre_beta45}
 }
 \caption{The optimal threshold $\Delta_\beta^\star$ vs transmission error probabilities with different values of $\beta$}
 \label{fig:OptThre}
\end{figure*}
\else
\begin{figure*}[]
 \subfloat[$\beta=5$]{
 \includegraphics[scale=.28]{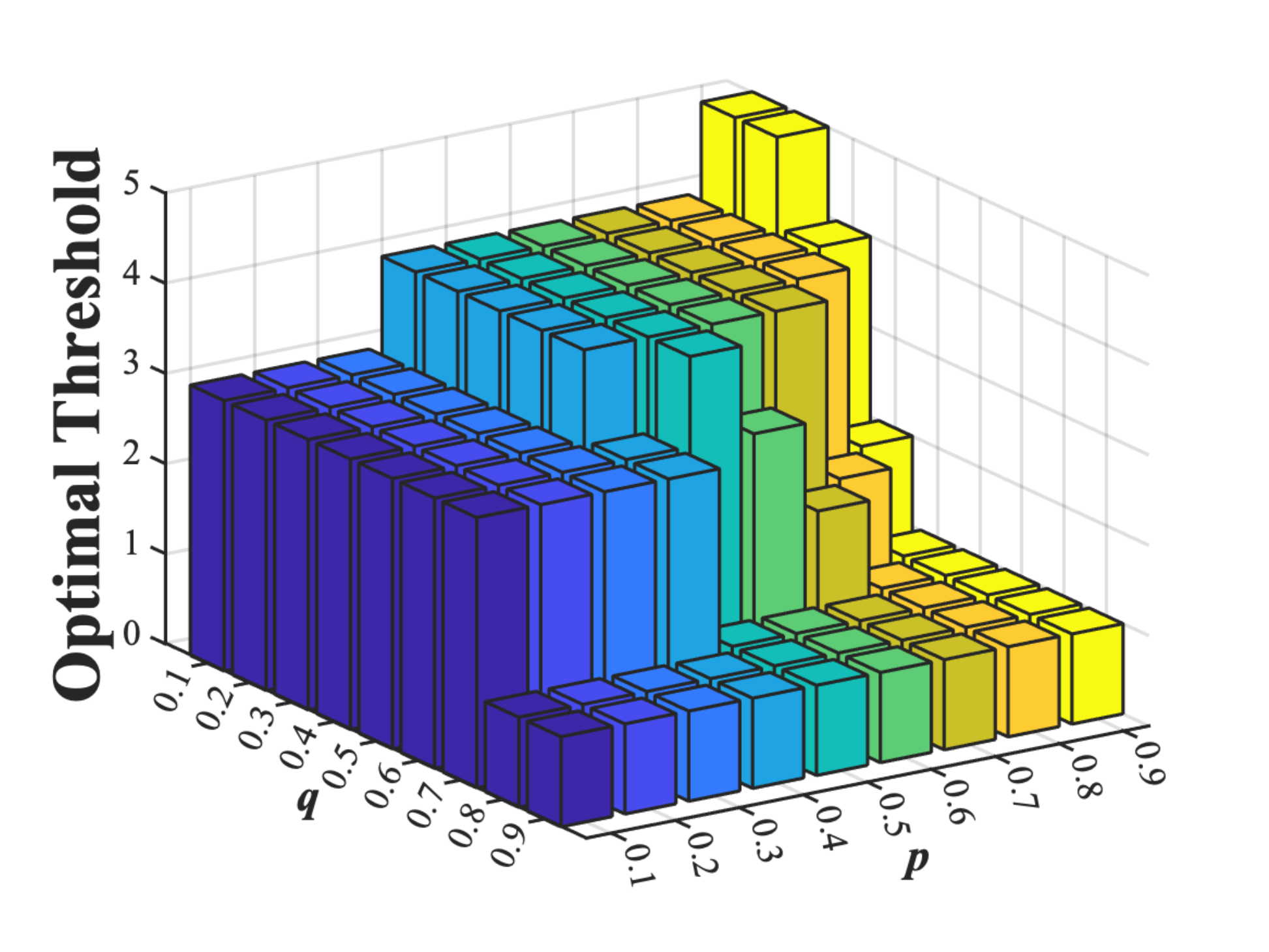}
 \label{fig:thre_beta5}
 }\ 
 \subfloat[$\beta=25$]{
 \includegraphics[scale=.28]{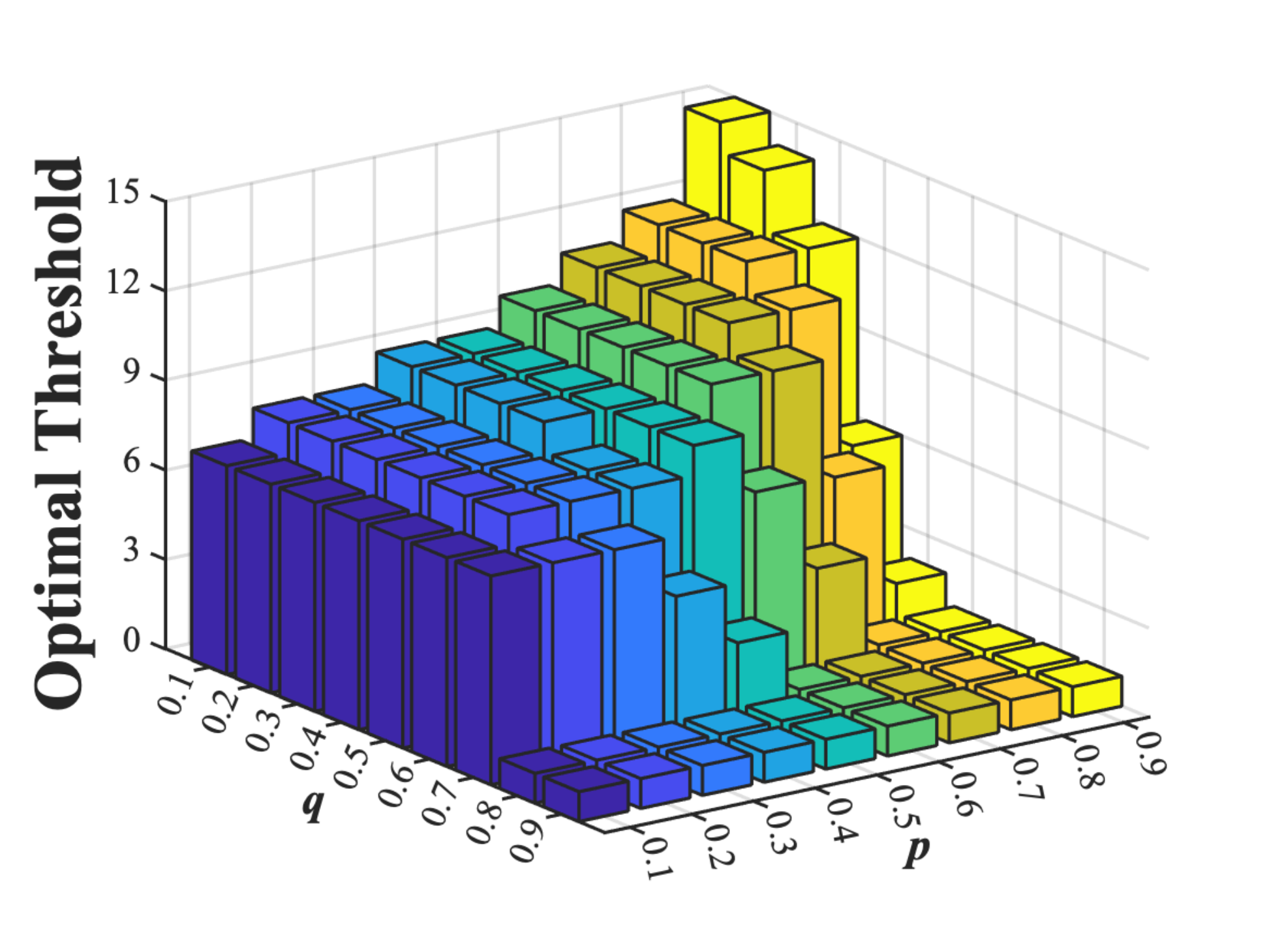}
 \label{fig:thre_beta25}
 }\  
 \subfloat[$\beta=45$]{
 \includegraphics[scale=.28]{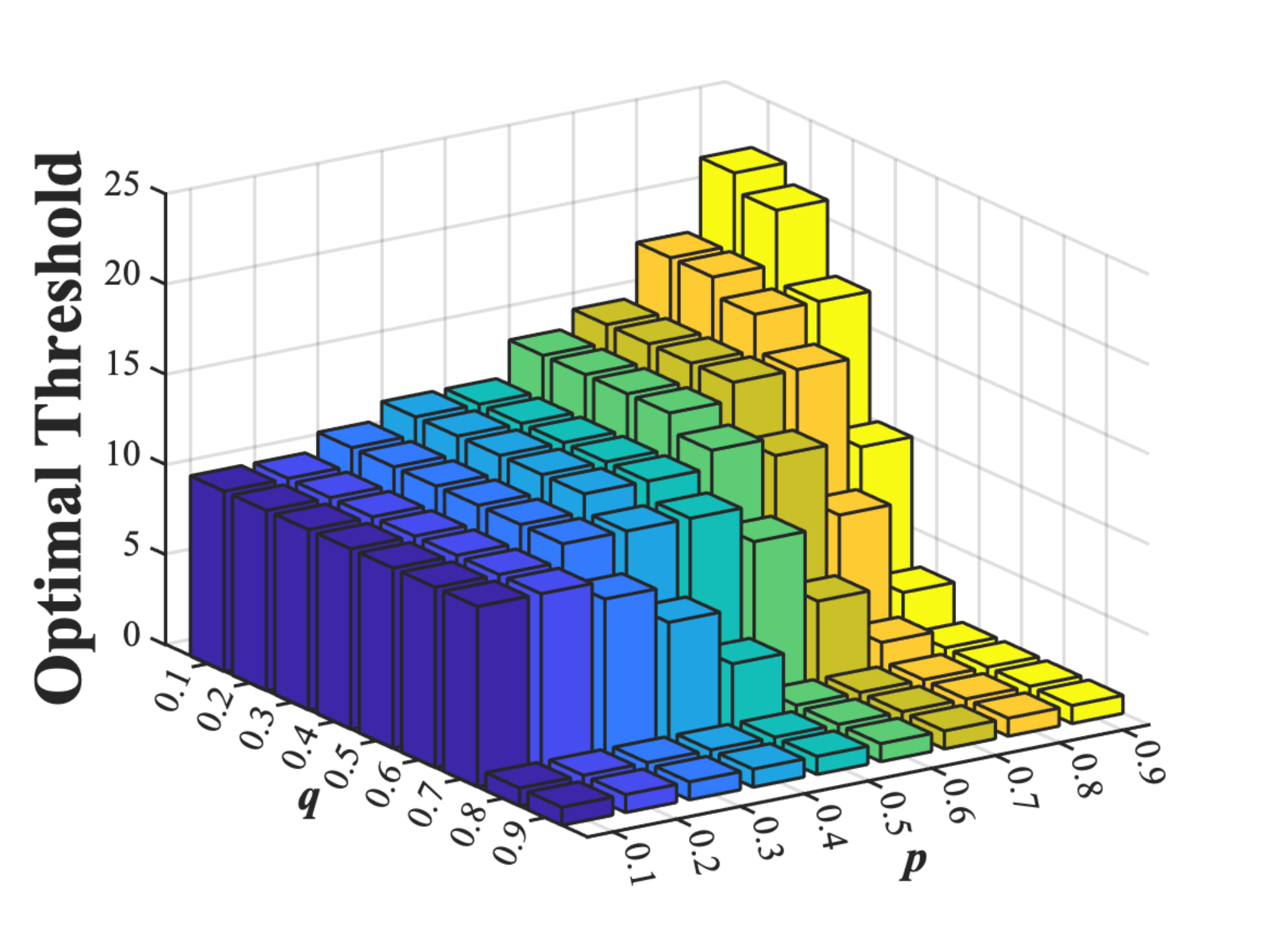}
 \label{fig:thre_beta45}
 }
 \caption{The optimal threshold $\Delta_\beta^\star$ vs transmission error probabilities with different values of $\beta$}
 \label{fig:OptThre}
\end{figure*}
\fi 

\vspace{-0.2cm}
\subsection{Comparison with greedy policy}
Let $e_t$ denote the total energy consumption before the time slot $t$. Then, $\bar{e}_t=e_t/(t-1)$ denotes the average energy cost consumed before $t$. In this part, we compare the Algorithm \ref{alg1} with a greedy policy which transmits 
whenever transmission is allowed (i.e., $\Lambda\geq D(\Delta)$) and if the empirical energy cost is less than the energy budget (i.e., $\bar{e}_t<E_\text{max}$). 
The setting is same as in 2) of \ref{sec:simu1}. 
\ifreport
\begin{figure}
    \centering
    \includegraphics[width=0.65\textwidth]{Age_dist/comp_greedy.eps}
    \caption{Average age vs energy constraint $E_{\text{max}}$ using the proposed policy (Algorithm \ref{alg1}) and greedy policy}
    \vspace{-0.7cm}
    \label{fig:comp_greedy}
\end{figure}
\else
\begin{figure}
    \centering    \includegraphics[width=0.41\textwidth]{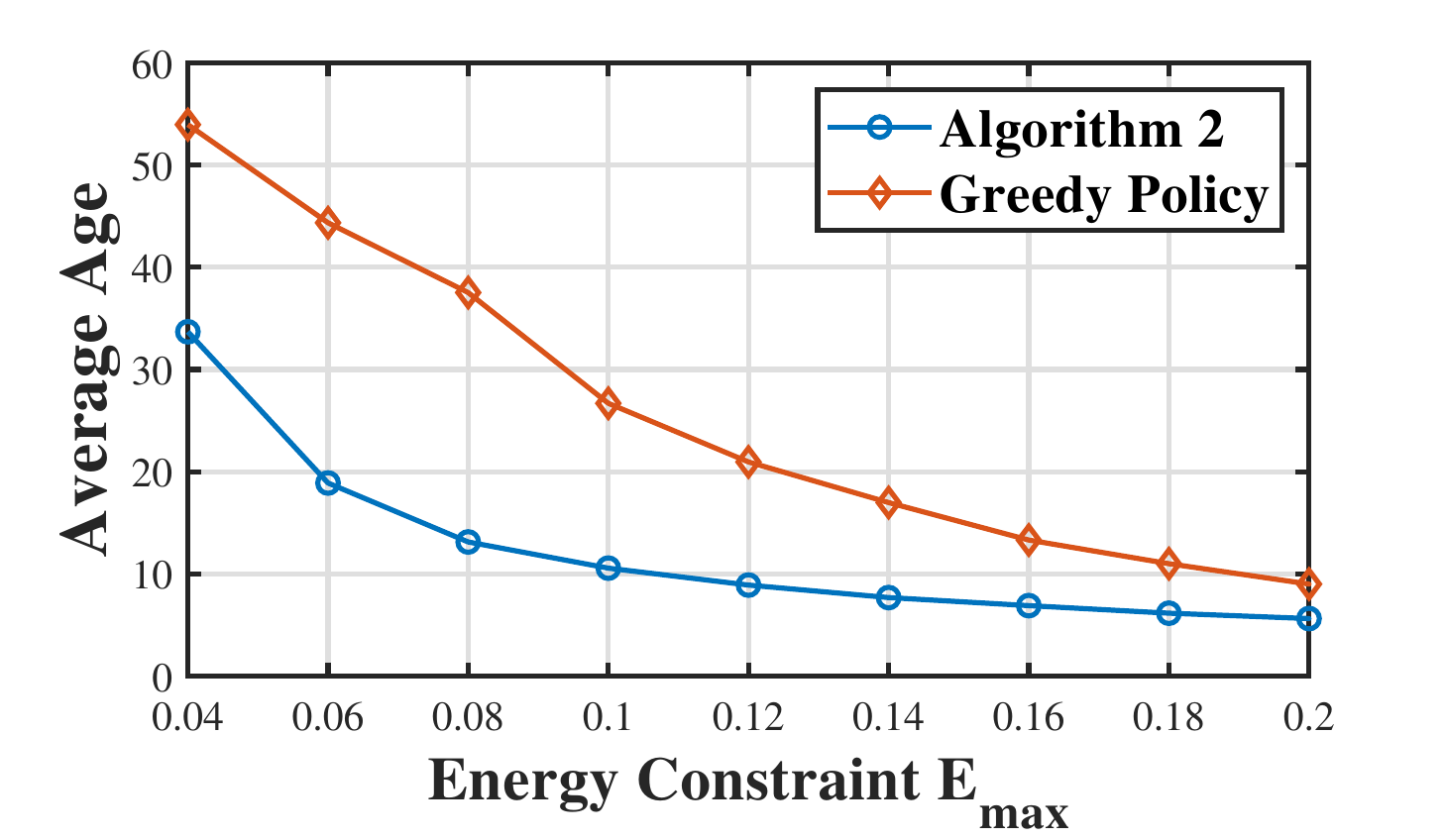}
    \caption{Average age vs energy constraint $E_{\text{max}}$ using the proposed policy (Algorithm \ref{alg1}) and greedy policy}
    \label{fig:comp_greedy}
         \vspace{-0.7cm}
\end{figure}
\fi 
When $E_{\text{max}}=1$, zero-waiting policy is obviously optimal and thus is considered uninteresting. In practical application like remote health, the transmission is likely to be only a small fraction of the total time duration because of the excessive energy consumption. This makes tight energy constraint more interesting to study. 
In Fig. \ref{fig:comp_greedy}, we 
observe that when $0.04\leq E_{\text{max}}\leq 0.2$, the improvement of our policy is quite significant, $30\%$ to $70\%$ reduction of cost. 


\section{Conclusion}
\ifreport
In this chapter, 
\else
In this paper, 
\fi 
we investigate an age minimization problem with constraints on the long-term average energy consumption and distortion of each update. The problem is formulated as a constrained MDP. Through the Lagrangian multiplier technique, we connect the problem to an average cost problem \eqref{avg_cost}, and show that the optimal policy is a mixture of two stationary deterministic policies (Theorem \ref{Ori_Lag}), each of which is optimal for the average cost problem and of a threshold-type (Theorem \ref{opt_stru}). Then, we obtain the average cost under the threshold-type policy, which is a piecewise function of threshold (Theorem \ref{cost_cal}), and we find the optimal threshold value in the last interval (Theorem \ref{new}). With these, we avoid dealing with infinite state space when using a classical solution RVI, and develop low-complexity algorithms. In the special, but practically very important case of constant distortion requirements, we obtain a closed-form solution (Corollary \ref{cor: special_case}). We show that the optimal threshold increases with the probability that distortion requirement is met, and the impact of transmission error probability $p$ on the optimal threshold depends on whether it is age or energy being the dominant issue (Theorem \ref{lem:special_case}).






%


\bibliographystyle{unsrt}
\bibliography{Reference}
\appendices
\section{Proof of Theorem \ref{Ori_Lag}}
\label{proof_ori_lag}
First, we provide some definitions of terms which will be used in our proof. The definitions comply with \cite{sennott1993constrained}: Let $G\subset \mathcal{S}$ be nonempty set of states. Give $\mathbf{s}\in\mathcal{S}$, $\mathcal{R}(\mathbf{s},G)$ is defined as a class of policies such that $\mathbb{P}^\pi(\mathbf{s}_t\in G\, \text{for some}\, t\geq 1: \mathbf{s}_0=\mathbf{s})=1$ and the expected time $\tau_{\mathbf{s},G}(\pi)$ of the first passage from $\mathbf{s}$ to $G$ using $\pi$ is finite. Further, $\mathcal{R}^\star(\mathbf{s},G)\subset\mathcal{R}(\mathbf{s},G)$ are policies that have finite expected average AoI and finite expected energy of a first passage from $\mathbf{s}$ to $G$.
By \cite{sennott1993constrained}, it suffices to show the following conditions hold.
\begin{itemize}
    \item A1: For all $r>0$, the set $B(r)=\{\mathbf{s}: \exists\, u\,\, \text{s.t.} \,\,C_{\Delta}(\mathbf{s},u)+C_{E}(\mathbf{s},u)\leq r \}$ is finite.
    \item A2: There exists a stationary deterministic policy $g_1\in\Pi$ which induces a Markov chain with properties: the state space incurred by $g_1$ consists of a single (non-empty) positive recurrent class $\mathcal{X}$ and a set $\mathcal{Y}$ of transient states such that $g_1\in\mathcal{R}^\star(\mathbf{s}',\mathcal{X})$, for $\mathbf{s}'\in \mathcal{Y}$, Moreover, both the average age and energy costs on $\mathcal{X}$ are finite.
    \item A3: Given any two states $\mathbf{s}_1\neq\mathbf{s}_2$, there exists a policy $g_2$ such that $g_2\in\mathcal{R}^\star(\mathbf{s}_1,\{\mathbf{s}_2\})$
   \item A4: If a stationary deterministic policy $g_3$ has at least one positive recurrent state then it has a single positive recurrent class $\mathcal{X}$. Moreover, if initial state $\mathbf{s}_0\notin \mathcal{X}$, then $g_3\in\mathcal{R}^\star(\mathbf{s}_0,\mathcal{X})$.
 \item A5: There exists a policy $g_4$ such that $\bar{A}(g_4)<\infty$ and $\bar{E}(g_4)< E_{\text{max}}$.
\end{itemize}
For A1, given $r$, for any $\mathbf{s}'=(\Delta',\Lambda')\in B(r),\Delta'\leq \min_u\{C_{\Delta}(\mathbf{s}',u)+C_{E}(\mathbf{s}',u)\}\leq r$ by definition of $B(r)$. This means given $r$, the age of any state in $B(r)$ is upper bounded by $r$. Together with $\Delta\in\mathbb{N}^+$ and $\Lambda\in\{0,1,\cdots,M\}$, $B(r)$ is finite.


For A2, consider policy $g_1$ that takes action $u=1$ if $\Lambda=M$; otherwise, $u=0$. The set $\mathcal{X}=\{(\Delta,\Lambda):\Delta\in\mathbb{N}^+,\Lambda'=0,1,...,M\}$ is recurrent. This is because that the next state after $u=1$ is $(o,\Lambda')\in \mathcal{X}$ due to $0<p<1$, and the next state after $u=0$ is $(o+1,\Lambda')\in \mathcal{X}$, where $\Lambda'\in\{0,1,...,M\}$, $o\in\mathbb{N}^+$. Thus, $\mathcal{Y}=\emptyset$ and $g_1\in\mathcal{R}^\star(\mathbf{s}',\mathcal{X})$, given $\mathbf{s}'\in \mathcal{Y}$. Besides, $\bar{A}(g_1)=1/p_s$, $\bar{E}(g_1)\leq 1$ are finite, where $p_s=P_\Lambda(M)(1-p)$. Hence, A2 holds.

For A3, given $\mathbf{s}_1=(\Delta_1,\Lambda_1)$ and $\mathbf{s}_2=(\Delta_2,\Lambda_2)$, consider the policy $g_2$ that uses $g_1$ in the proof of A2 till entry to state $(1,\cdot)$, and then takes $u=0$ till $\Delta=\Delta_2$, after which policy repeats previous two stages. Based on analysis in proof of A2, it takes finite time to enter state $(1,\cdot)$ from $\mathbf{s}_1$. In addition, it takes $\Delta_2-1$ slots to reach the age $\Delta_2$. Thus, the two stages take finite time. Moreover, the probability that $\mathbf{s}\neq\mathbf{s}_2$ at the end of the two stage exponentially decreases with the number of the two stage being conducted. Hence, A3 holds.

For A4, the only way for $g_3$ to generate at least one recurrent class is that successful transmission occurs repeatedly under $g_3$. Note after successful transmission, state becomes $(1,\Lambda)$, $\Lambda\in\{0,1,\cdots,M\}$. Thus, any recurrent class must include $(1,\Lambda)$. Hence, there is only one recurrent class. Moreover, since $g_3$ will take $u=1$ repeatedly, for any initial state $\mathbf{s}_0\notin \mathcal{X}$, it takes finite time from $\mathbf{s}_0$ to $(1,\Lambda)$. Hence, A4 holds.

For A5, consider policy $g_4$ that takes action $u=1$ if $2\lceil 1/E_\text{max}\rceil$ divides $\Delta$ and $\Lambda=M$; otherwise, $u=0$. We have $\bar{E}(g_4)\leq 1/2\lceil 1/E_\text{max}\rceil \leq E_\text{max}/2<E_\text{max}$ and $\bar{A}(g_4)=\lceil 1/E_\text{max}\rceil+2\lceil 1/E_\text{max}\rceil^2(2/p_s-1)<\infty$, where $p_s$ is defined in proof of A2. This completes our proof.

\section{Proof of Theorem \ref{cost_cal}}
\label{opt_thres0}

The state transition diagram under the policy in the form of \eqref{optPol} is given in Fig. \ref{fig:trans_diag}, where $k$ is a threshold.

\begin{figure*}
    \centering
    \includegraphics[width=0.95\textwidth]{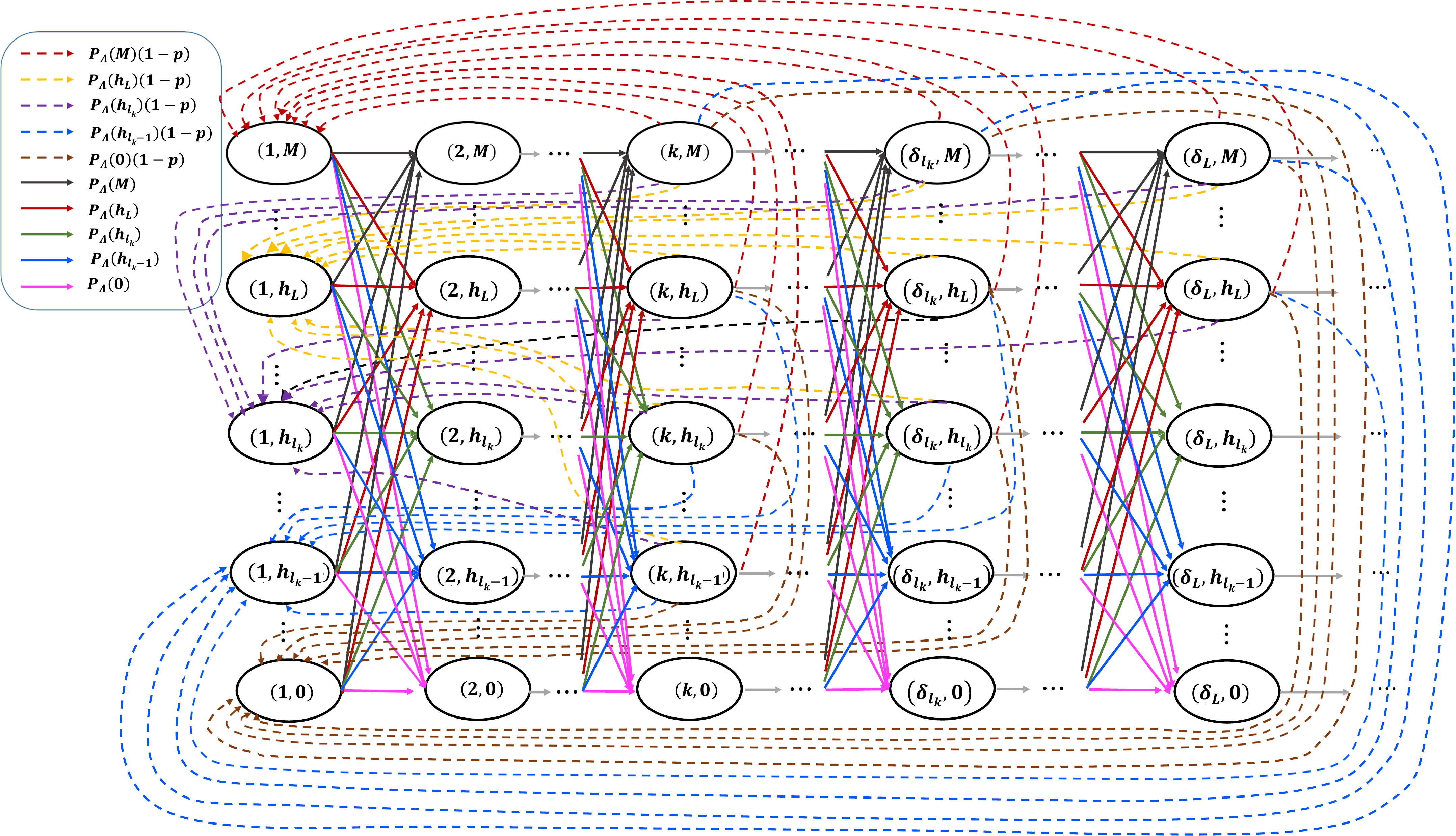}
    \caption{State Transition Diagram. Different lines correspond to different state transition probabilities as given in the legend.}
    \label{fig:trans_diag}
\end{figure*}

Define the steady state probability as follows:
\begin{align}
    x_{\Delta,\Lambda}\triangleq \mathbb{P}(\mathbf{s}=(\Delta,\Lambda))
\end{align}
Based on the state transition diagram, balance equations can be obtained as follows:
\begin{align}
 \! \! & x_{\Delta+1,\Lambda}\!=\!\sum_{b=0}^M x_{\Delta,b} P_\Lambda(\Lambda), & \forall 1\!\leq\! \Delta\!<\! k\\
 \! \! & x_{\Delta+1,\Lambda}\!=\!(\sum_{b=0}^M x_{\Delta,b}\!-\!(1\!-\! p)\!\!\!\sum_{b=h_{{l_k}-1}}^M \!\!x_{\Delta,b})P_\Lambda(\Lambda), & \forall k\!\leq \!\Delta\!<\!\delta_{l_k}\\
 \! \! &  \cdots \notag\\
  \! \! & x_{\Delta+1,\Lambda}\!=\!(\sum_{b=0}^M x_{\Delta,b}\!-\!(1\!-\!p)\!\!\sum_{b=h_L}^M \!\!x_{\Delta,b})P_\Lambda(\Lambda), & \forall  \delta_L\!\leq \!\Delta\\
\! \!  &\sum_{\Delta=1}^\infty \sum_{\Lambda=0}^M x_{\Delta,\Lambda}=1 &
\end{align}
Let $z_\Delta=\sum_{\Lambda=0}^M x_{\Delta,\Lambda}$. Then, the balance equations can be transformed as
\begin{align}
   & z_{\Delta+1}=z_\Delta & \forall 1\leq \Delta< k\label{be1}\\
   & z_{\Delta+1}=z_\Delta B_{{l_k}-1} & \forall k\leq \Delta<\delta_{l_k}\label{be2}\\
   &\ \ \ \vdots\notag\\
   & z_{\Delta+1}=z_\Delta B_{L} & \forall \delta_L\leq \Delta\\
   & \sum_{\Delta=1}^\infty z_\Delta=1 &\label{be3}
\end{align}
where $B_{r}=1-(1-p)\sum_{\Lambda=h_{r}}^M P_\Lambda(\Lambda)$, for ${l_k}-1\leq r\leq L$. 

Solving the equations \eqref{be1}-\eqref{be3}, we obtain 
\begin{align}
    z_1=&\Big[k+\frac{B_{{l_k}-1}}{1-B_{{l_k}-1}}
    +\mathbbm{1}_{\{{l_k}\leq L\}}B_{{l_k}-1}^{\delta_{l_k}-k}\Big(-\frac{1}{1-B_{{l_k}-1}}\notag\\
    &+\sum_{j={l_k}}^{L}\frac{1-\mathbbm{1}_{\{j<L\}}B_{j}^{\delta_{j+1}-\delta_j}}{1-B_j}w({l_k},j)\Big)\Big]^{-1},
\end{align}
where $\mathbbm{1}_{\{\cdot\}}$ is indicator function and $w(i,j)$ is defined as
\begin{align}
   w(i,j)=\mathbbm{1}_{\{i<j\}}\prod_{v=i}^{j-1}B_{v}^{\delta_{v+1}-\delta_v}+\mathbbm{1}_{\{i\geq j\}},
\end{align}
and for $\Delta>1$, $z_\Delta$ is given by
\begin{align}
    & z_\Delta=z_1,   & \forall 1<\Delta<k\\
    & z_\Delta=z_1 B_{{l_k}-1}^{\Delta-k}, &\forall k\leq \Delta<\delta_{j+1} \\
    & z_\Delta=w({l_k},j)B_{{l_k}-1}^{\delta_{l_k}-k}B_j^{\Delta-\delta_j}, &\forall \delta_j\leq \Delta<\delta_{j+1} \\
    & z_\Delta=w({l_k},L)B_{{l_k}-1}^{\delta_{l_k}-k}B_L^{\Delta-\delta_L}, &\forall \delta_L\leq \Delta
\end{align}

With this, the resultant average Lagrangian cost using the policy $\pi_k$ in the form of \eqref{optPol} with threshold $k\in\mathbb{N}^+$ is 
\begin{align}
    &\bar{L}(\pi_k;\beta)\notag\\
    =&\sum_{\Delta=1}^{k-1}\Delta z_\Delta+\sum_{\Delta=k}^\infty \Big(\Delta+\beta \sum_{b= D(\Delta)}^{M}P_\Lambda(b)\Big)z_\Delta\\
    =&\sum_{\Delta=1}^{k-1}\Delta z_\Delta+\mathbbm{1}_{\{{l_k}=L+1\}}\cdot\sum_{\Delta=k}^\infty \Big(\Delta+\beta F(h_{{l_k}-1})\Big)z_k\notag\\
    &+\mathbbm{1}_{\{{l_k}\leq L\}}\Big(\mathbbm{1}_{\{{l_k}<L\}}\cdot\sum_{r={l_k}}^{L-1}\sum_{\Delta=\delta_r}^{\delta_{r+1}-1} (\Delta+\beta F(h_r))z_{\delta_r}\notag\\
    &+\sum_{\Delta=k}^{\delta_{l_k}-1} (\Delta+\beta F( h_{{l_k}-1}))z_k+\sum_{\Delta=\delta_L}^{\infty} (\Delta+\beta F(h_L))z_{\delta_L}\Big)\\
    =&z_1\Bigg\{\frac{k(k-1)}{2}+\frac{1}{(1-B_{{l_k}-1})^2}+\frac{k-1+\beta F(h_{{l_k}-1})}{1-B_{{l_k}-1}}\notag\\
    &+\mathbbm{1}_{\{{l_k}\leq L\}}B_{{l_k}-1}^{\delta_{l_k}-k}\bigg[-\frac{1}{(1-B_{{l_k}-1})^2}-\frac{\delta_{l_k}-1+\beta F(h_{{l_k}-1})}{1-B_{{l_k}-1}}\notag\\
    &+\sum_{j={l_k}}^{L}w({l_k},j)\Big(\frac{1-\mathbbm{1}_{\{j<L\}}B_{j}^{\delta_{j+1}-\delta_j}}{(1-B_j)^2}\notag\\
    &+\frac{-\mathbbm{1}_{\{j<L\}}B_{j}^{\delta_{j+1}-\delta_j}(\delta_{j+1}-1+\beta F(h_j))}{1-B_j}\notag\\
    &+\frac{\delta_j-1+\beta F(h_j)}{1-B_j}\Big)\bigg]\Bigg\}
\end{align}
where $F(h)\triangleq \sum_{b=h}^{M}P_\Lambda(b)$

And the average energy is
\begin{align}
    &\bar{E}(\pi_k)\notag\\
    =&\sum_{\Delta=k}^\infty \beta \sum_{b= D(\Delta)}^{M}P_\Lambda(b)z_\Delta\\
    =&z_1\Bigg\{\frac{\beta F(h_{{l_k}-1})}{1-B_{{l_k}-1}}+\mathbbm{1}_{\{{l_k}\leq L\}}B_{{l_k}-1}^{\delta_{l_k}-k}\bigg(-\frac{\beta F(h_{{l_k}-1})}{1-B_{{l_k}-1}}\notag\\
    &+\sum_{j={l_k}}^{L}w({l_k},j)\beta F(h_j)\frac{1-\mathbbm{1}_{\{j<L\}}B_{j}^{\delta_{j+1}-\delta_j}}{1-B_j}\bigg)\Bigg\}
\end{align}

\section{Proof of Theorem \ref{new}}
\label{proof_min_ub}
We use $\mathcal{F}(k)$ to denote the resultant cost expression $\bar{L}(\pi_k;\beta)$ for $k\geq \delta_L$. Then, $\mathcal{F}(k)$ is
\begin{align}
   &\mathcal{F}(k)=\notag\\
    &\frac{1-B_L}{B_L+k(1-B_L)}\left(\frac{k(k-1)}{2}+\frac{\beta F(h_{L})+k}{1-B_L}+\frac{B_L}{(1-B_L)^2}\right),
\end{align}

Next, we will find the minimum of $\mathcal{F}(k)$ for $k\geq \delta_L$.
To this end, we first show that $\mathcal{F}(k)$ firstly decreases and then increases with $k\in\mathbb{N}$. Then, we get optimal threshold on $[\delta_L,\infty)$ by comparing the $y$ that optimizes $\mathcal{F}(k)$ in the range $k\in\mathbb{N}$ with $\delta_L$.
Actually, after some basic calculation, we have
\begin{align}
  \mathcal{F}(k+1)-\mathcal{F}(k) = \frac{0.5k^2+(0.5+\frac{B_L}{1-B_L})k-\frac{\beta F(h_L)}{1-B_L}}{(k+1+\frac{B_L}{1-B_L})(k+\frac{B_L}{1-B_L})}\label{costfun}
\end{align}

Note that $1-B_L=(1-p)\sum_{b=h_L}^M P_\Lambda(b)>0$. Thus, the denominator is positive. Let $\mathcal{H}(k)=0.5k^2+(0.5+\frac{B_L}{1-B_L})k-\frac{\beta F(h_L)}{1-B_L}$. Since
$\mathcal{\delta}(k+1)-\mathcal{H}(k)=k+1+\frac{B_L}{1-B_L}>0$ for $k\geq 0$ and $\mathcal{H}(0)<0$, there exists $y\in\mathbb{N}$ such that $\mathcal{H}(k)<0$ for $k< y$ and $\mathcal{H}(k)\geq 0$ for $k\geq y$. Therefore, \eqref{costfun} is negative when $k<y$ and then becomes positive when $k\geq y$. This implies that $\mathcal{F}(k)$ firstly decreases and then increases with $k\in\mathbb{N}$, and $\mathcal{F}(y)$ is the minimum for $k\in \mathbb{N}$. Note that $y=\min\{k\in\mathbb{N}: \mathcal{H}(k)\geq 0\}$. 
Let $\sigma$ denote the solution to $\mathcal{H}(\sigma)=0$. Then, we have 
\begin{align}
    \sigma=-\frac{1+B_L}{2(1-B_L)}+\left(\frac{B_L^2}{(1-B_L)^2}+\frac{B_L+2\beta F(h_L)}{1-B_L}+\frac{1}{4}\right)^{0.5},
\end{align}
and $y=\lceil \sigma\rceil$.

If $y\leq \delta_L$, $\mathcal{F}$ increases with $k$ on the domain $[\delta_L,\infty)$ and the optimal threshold $k_{UB}$ in the domain is $\delta_L$; if $y> \delta_L$, $\mathcal{F}$ first decreases and increases with $k$ on the domain $[\delta_L,\infty)$, and thus the optimal threshold $k_{UB}$ is $y$. Hence, the optimal threshold on the domain $[\delta_L,\infty)$, which is denoted by $k_{UB}$, is
\begin{equation}
    k_{UB} = \max \{\delta_L,y\}
\end{equation}

\section{Proof of Theorem \ref{lem:special_case}}
\label{app: special_case}
To explore how the optimal threshold varies with $W$, $p$, respectively. We regard the optimal threshold as a function of $W$ and $p$, i.e.,
\begin{align}
    \mathcal{G}(W,p)\triangleq &-\frac{1+R}{2(1\!-\!R)}\!+\!\left(\!\frac{R^2}{(1\!-\!R)^2}\!+\!\frac{R\!+\!2\beta W}{1\!-\!R}\!+\!\frac{1}{4}\right)^{0.5},\\
    = & \frac{-2+(1-p)W}{2(1-p)W}+\!\Bigg(\!\frac{(1-(1-p)W)^2}{(1-p)^2W^2}\!\notag\\
    &+\!\frac{1-(1-p)W\!+\!2\beta W}{(1-p)W}\!+\!\frac{1}{4}\Bigg)^{0.5},
\end{align}
and then study partial derivatives $\frac{\partial \mathcal{G}(W,p)}{\partial W}$ and $\frac{\partial \mathcal{G}(W,p)}{\partial p}$.

(i) We have $\frac{\partial \mathcal{G}(W,p)}{\partial p}$ as
\begin{align}
    &\frac{\partial \mathcal{G}(W,p)}{\partial p}=\notag\\
    &\frac{1}{(1-p)^2W}\Bigg(-1+\frac{2-(1-p)W+2\beta(1-p)W^2}{\sqrt{(2-(1-p)W)^2+8\beta(1-p)W^2}}\Bigg)
    \label{eqn:der_p}
\end{align}
Since
\begin{align}
    &\big(2-(1-p)W\big)^2+8\beta(1-p)W^2\notag\\
    &-\big(2-(1-p)W+2\beta(1-p)W^2\big)^2\notag\\
    =& 4\beta(1-p)^2W^3(1-\beta W)
\end{align}
If $W >\frac{1}{\beta}$, then \eqref{eqn:der_p}$>0$. In the case, the $\mathcal{G}(W,p)$ increases with $p$. If $ W <\frac{1}{\beta}$, then \eqref{eqn:der_p}$<0$. In the case, the $\mathcal{G}(W,p)$ decreases with $p$. 

(ii) It suffices to show that the maximum of
$$\Bigg\lceil\!\!-\!\frac{1+R}{2(1\!-\!R)}\!+\!\left(\!\frac{R^2}{(1\!-\!R)^2}\!+\!\frac{R\!+\!2\beta W}{1\!-\!R}\!+\!\frac{1}{4}\right)^{0.5}\!\!\Bigg \rceil$$ in \eqref{opt_threval_spec} is not larger than 1. In fact, when $\beta <\frac{1}{W}$ and $0<R<1$, we have 
\begin{align}
    \frac{R^2}{(1\!-\!R)^2}\!+\!\frac{R\!+\!2\beta W}{1\!-\!R}\!+\!\frac{1}{4}&<\frac{R^2}{(1\!-\!R)^2}\!+\!\frac{R\!+2}{1\!-\!R}\!+\!\frac{1}{4}\\
    &\leq\left(1+\!\frac{1+R}{2(1\!-\!R)}\right)^2.
\end{align}
This completes the proof of (ii).

For (iii), we show that $\frac{\partial \mathcal{G}(W,p)}{\partial W}\geq 0$.
(iii) It suffices to show that $\frac{\partial \mathcal{G}(W,p)}{\partial W}\geq 0$.
In fact, we have
\begin{align}
    &\frac{\partial \mathcal{G}(W,p)}{\partial W}=\notag\\
    &\frac{1}{(1-p)W^2}\Bigg(1+\frac{-2+(1-p)W}{\sqrt{(2-(1-p)W)^2+8\beta(1-p)W^2}}\Bigg)
\end{align}
Since $0\!<\!\!\frac{2-(1-p)W}{\sqrt{(2-(1-p)W)^2+8\beta(1-p)W^2}}\!< \!1$, we have $\frac{\partial \mathcal{G}(W,p)}{\partial W}\!\!>\!0$. 

\section{Proof of Lemma \ref{property_00}}
\label{proof_prop_dis_subprob}
By Proposition \ref{existence_discount_00}, we will use induction to show the results. Obviously, $V_0^{\alpha}(\mathbf{s})$ has properties (i) and (ii) since $V_0^{\alpha}(\mathbf{s})=0$. Then, it remains to show that given $V_n^{\alpha}(\mathbf{s})$ has the properties (i) and (ii), $V_{n+1}^{\alpha}(\mathbf{s})$ has these properties, $\forall n\geq 0$.

(i) Let $\Delta_1\geq \Delta_2$. By \eqref{iteration00}, to show that the result holds for $V_{n+1}^{\alpha}$, it suffices to show that for any $u\in A_{(\Delta_1,\Lambda)}$, there exists an action $u'\in A_{(\Delta_2,\Lambda)}$ such that $Q_{n+1}^{\alpha}(\Delta_1,\Lambda;u)\geq Q_{n+1}^{\alpha}(\Delta_2,\Lambda;u')$.

For any state, $0\in A_{(\Delta_1,\Lambda)}$. We have
\begin{align}
   &Q_{n+1}^{\alpha}(\Delta_1,\Lambda;0)\nonumber\\
   =&\Delta_1+\alpha\mathbb{E}[ V_n^{\alpha}(\Delta_1+1,\Lambda')]\\
  =&\Delta_1+\alpha\sum_{k=0}^{M} P_\Lambda(k)V_n^{\alpha}(\Delta_1+1,k)\\
  \geq&\Delta_2+\alpha\sum_{k=0}^{M} P_\Lambda(k)V_n^{\alpha}(\Delta_2+1,k)\label{appD_1}\\
  =& \Delta_2+\alpha\mathbb{E}[ V_n^{\alpha}(\Delta_2+1,\Lambda')]\\
  =&Q_{n+1}^{\alpha}(\Delta_2,\Lambda;0)
\end{align}

The inequality \eqref{appD_1} holds by our assumption that $V_n^{\alpha}$ has non-decreasing property. 

Since $D(\Delta_1)\geq D(\Delta_2)$, if $1\in A_{(\Delta_1,\Lambda)}$, then $1\in A_{(\Delta_2,\Lambda)}$. In this case, we have
\begin{align}
   &Q_{n+1}^{\alpha}(\Delta_1,\Lambda;1)\nonumber\\
   =&\Delta_1+\beta+\alpha p \mathbb{E}[ V_n^{\alpha}(\Delta_1+1,\Lambda')]\nonumber\\
   &+\alpha (1-p) \mathbb{E}[ V_n^{\alpha}(1,\Lambda')]\\
   \geq & \Delta_2+\beta+\alpha p \mathbb{E}[ V_n^{\alpha}(\Delta_2+1,\Lambda')]\nonumber\\
   &+\alpha (1-p) \mathbb{E}[ V_n^{\alpha}(1,\Lambda')]\\
   =&Q_{n+1}^{\alpha}(\Delta_2,\Lambda;1)
\end{align}

(ii) Let $\Lambda_1\leq \Lambda_2$.
\begin{align}
   Q_{n+1}^\alpha(\Delta,\Lambda_1;0)&=\Delta+\alpha\mathbb{E}[ V_n^\alpha(\Delta+1,\Lambda')]\nonumber\\
  &= Q_{n+1}^\alpha(\Delta,\Lambda_2;0)
\end{align}
If $\Lambda_1\geq D(\Delta)$, then $\Lambda_2\geq D(\Delta)$.
\begin{align}
   Q_{n+1}^\alpha(\Delta,\Lambda_1;1)=&\Delta+\beta+\alpha p \mathbb{E}[ V_n^\alpha(\Delta+1,\Lambda')]\nonumber\\
   &+\alpha (1-p) \mathbb{E}[ V_n^\alpha(1,\Lambda')]\nonumber\\
   =&Q_{n+1}^\alpha(\Delta,\Lambda_2;1)
\end{align}
By \eqref{iteration00}, $V_{n+1}^\alpha$ has property (ii).

\section{Proof of Lemma \ref{disc_stru0}}
\label{proof_disc_stru}
To show the result, we first show that given $\Lambda$, the optimal action is increasing function of the age when the age satisfies $\Lambda\geq D(\Delta)$, and then show that the optimal action at certain age is the same for any $\Lambda$ as long as distortion requirement is satisfied.

(i) We show that given $\Lambda$, if $u=1$ is optimal for the state $(\Delta,\Lambda)$, then $u=1$ is also optimal for state $(\Delta+1,\Lambda)$, where $\Delta+1\leq D^{-1}(\Lambda)$.

Since $u=1$ is optimal for $(\Delta,\Lambda)$, we have 
\begin{align}
    &Q^{\alpha}(\Delta,\Lambda;1)-Q^{\alpha}(\Delta,\Lambda;0)\nonumber\\
    =&\beta\!+\!\alpha (1\!-\!p)(\mathbb{E}V^{\alpha}(1,\Lambda')\!-\!\mathbb{E}V^{\alpha}(\Delta+1,\Lambda'))\label{thre_01}\leq 0
\end{align}
By Lemma \ref{property_00}, we have 
\begin{align}
    &Q^{\alpha}(\Delta+1,\Lambda;1)-Q^{\alpha}(\Delta+1,\Lambda;0)\nonumber\\
    =&\beta\!+\!\alpha (1\!-\!p)(\mathbb{E}V^{\alpha}(1,\Lambda')\!-\!\mathbb{E}V^{\alpha}(\Delta+2,\Lambda'))\\
    \leq & \beta\!+\!\alpha (1\!-\!p)(\mathbb{E}V^{\alpha}(1,\Lambda')\!-\!\mathbb{E}V^{\alpha}(\Delta+1,\Lambda'))\leq 0
\end{align}

(ii) Next, we show that for $\Lambda<M$ and $\Delta\leq D^{-1}(\Lambda)$, $u^\star_{\alpha}(\Delta,\Lambda;\beta)=u^\star_{\alpha}(\Delta,M;\beta)$, where $u^\star_{\alpha}(\cdot;\beta)$ is optimal decision rule.

Since for any $\Lambda\in \{0,1,\cdots,M\}$, we have
\begin{align}
    &Q^{\alpha}(\Delta,\Lambda;1)-Q^{\alpha}(\Delta,\Lambda;0)\nonumber\\
    =&\beta\!+\!\alpha (1\!-\!p)(\mathbb{E}V^{\alpha}(1,\Lambda')\!-\!\mathbb{E}V^{\alpha}(\Delta+1,\Lambda')),
\end{align}
which does not depend on $\Lambda$. Thus,
$Q^{\alpha}(\Delta,\Lambda;1)-Q^{\alpha}(\Delta,\Lambda;0)=Q^{\alpha}(\Delta,M;1)-Q^{\alpha}(\Delta,M;0)$. Hence, $u^\star_{\alpha}(\Delta,\Lambda;\beta)=u^\star_{\alpha}(\Delta,M;\beta)$.

\section{Proof for verification of conditions in \cite{sennott1989average}}
\label{proof_cond}
We need to verify the conditions listed below:
	\begin{itemize}
		\item A1: $V^{\alpha}(\mathbf{s})$ defined in \eqref{disc_cost_opt_00} is finite $\forall \mathbf{s},\alpha$.
		\item A2: $\exists I\geq 0$ s.t. $-I\leq h^{\alpha} (\mathbf{s})\triangleq V^{\alpha}(\mathbf{s})-V^{\alpha}(\mathbf{0})$, $\forall \mathbf{s}, \alpha$.
		\item A3: $\exists F(\mathbf{s})\geq 0$ s.t. $h^{\alpha} (\mathbf{s})\leq F(\mathbf{s})$, $\forall \mathbf{s},\alpha$. Moreover, for each $\mathbf{s}$, $\exists \, u(\mathbf{s})$ s.t. $\sum_{\mathbf{s'}\in \mathcal{S}}\mathbb{P}(\mathbf{s'}|\mathbf{s},u(\mathbf{s}))F(\mathbf{s'})<\infty$.
		\item A4: $\sum_{\mathbf{s'}\in \mathcal{S}}\mathbb{P}(\mathbf{s'}|\mathbf{s},u)F(\mathbf{s'})<\infty$ $\forall \mathbf{s}, u$.
	\end{itemize}
	
In Proposition \ref{existence_discount_00}, we showed that a policy $f$ that chooses $u=0$ at every time slot satisfies $L_{\mathbf{s}}^{\alpha}(f;\beta)\!<\!\!\infty$.
By \eqref{disc_cost_opt_00}, we have $L_{\mathbf{s}}^{\alpha}(f;\beta)\geq V^{\alpha}(\mathbf{s})$, which implies A1.

By Lemma \ref{property_00}, we have $V^{\alpha}(\Delta',\Lambda)\geq V^{\alpha}(\Delta,\Lambda)$ for all $\Delta'\geq \Delta$ and $V^{\alpha}(\Delta,\Lambda)\geq V^{\alpha}(\Delta,\Lambda')$ for all $\Lambda'\geq \Lambda$. Hence, by setting $I=V^{\alpha}(\mathbf{0})-V^{\alpha}(1,M)$, where $\mathbf{0}=(1,M)$ is the reference state, we prove A2.

Let $g$ be the policy that transmits whenever the number of collected samples equals $M$. This ensures that any transmission satisfies distortion requirement. Under policy $g$, states that occur after successful delivery are recurrent. Actually, the probability that no transmission succeeds after $l$ slots is $(1-P_\Lambda(M)(1-p))^{l}$. State $\mathbf{0}$ follows a successful delivery and is recurrent.
 Hence, under policy $g$ the expected cost of the first passage from state $\mathbf{s}$ to $\mathbf{0}$, denoted by $c_{\mathbf{s},\mathbf{0}}(g)$, is finite.
 Let $g'$ be a mix policy where $g$ is used until entering state $\mathbf{0}$ and the optimal policy for the $\alpha$-discounted cost problem, denoted by $g_{\alpha}$, is used afterwards. Suppose $T$ is the first time slot when system enters $\mathbf{0}$. Then, we have,
   \begin{align}
   V^{\alpha}(\mathbf{s})
   	\leq &\mathbb{E}_{g'}[\sum_{t=1}^{T-1}\!\!\alpha^{t-1}C(\mathbf{s}_t,\!u_t;\beta)|\mathbf{s}]\nonumber\\
   	&+\!\mathbb{E}_{g'}[\sum_{t=T}^{\infty}\!\!\alpha^{t-1}C(\mathbf{s}_t,\!u_t;\beta)|\mathbf{0}]\\
   	\leq & c_{\mathbf{s},\mathbf{0}}(g)+ \mathbb{E}_{g_{\alpha}}(\alpha^{(T-1)})V^{\alpha}(\mathbf{0})\\
   	\leq & c_{\mathbf{s},\mathbf{0}}(g)+ V^{\alpha}(\mathbf{0}).
   \end{align}  
   Hence, by setting $F(\mathbf{0})=0$ and $F(\mathbf{s})=c_{\mathbf{s},\mathbf{0}}(g)$ for $\mathbf{s}\!\neq \!\mathbf{0}$, we prove A3. After transition from $\mathbf{s}$ under any action, there will be at most two possible states. Since for all $\mathbf{s}$, $F(\mathbf{s})<\infty$, the sum of at most two $F(\cdot)$ is also finite. Hence, A4 holds.
\end{document}